\documentclass[journal]{IEEEtran}
\usepackage{amsmath}
\usepackage{eqparbox}
\usepackage[usenames, dvipsnames]{color}
\usepackage{mathtools,amssymb,bm,mathabx}
\usepackage{amstext}
\usepackage{amssymb}
\usepackage{graphicx}
\usepackage{color}
\usepackage{booktabs}
\usepackage{longtable}
\usepackage{multicol}
\usepackage{amsfonts}
\usepackage{dsfont}
\usepackage{array}
\usepackage[most]{tcolorbox}
\usepackage[ruled]{algorithm}
\usepackage{algorithmic}
\usepackage{cases}
\usepackage{pgfplots}
\usepackage{caption}
\usepackage{subfigmat}
\usepackage{xr-hyper}
\usepackage{hyperref}
\usepackage{amsthm}
\DeclareCaptionLabelSeparator{periodspace}{.\quad}
\usepackage{float}
\usepackage{cite}
\usepackage{epstopdf}
\theoremstyle{remark}

\newcommand\ASTART{\bigskip\noindent\begin{minipage}[b]{0.5\linewidth}}
	
	\newcommand\AENDSKIP{\end{minipage}\bigskip}
\newcommand\AEND{\end{minipage}}
\ifCLASSOPTIONcaptionsoff
\usepackage[nomarkers]{endfloat}
\let\MYoriglatexcaption\caption
\renewcommand{\caption}[2][\relax]{\MYoriglatexcaption[#2]{#2}}
\fi
\theoremstyle{plain}
\newtheorem{thm}{\textbf{Theorem}}
\newtheorem{lem}{\textbf{Lemma}}

\newtheorem{corl}{\textbf{Corollary}}

\theoremstyle{definition}

\theoremstyle{remark}
\newtheorem{rem}{\bf Remark}

\newcommand*{\rom}[1]{\expandafter\@slowromancap\romannumeral #1@}

\newcommand{\RN}[1]{%
\textup{\uppercase\expandafter{\romannumeral#1}}%
}

\usepackage{standalone}
\graphicspath{{./Figures/}}

\begin{document}
%
%\onecolumn

%\title{Blind Goal-oriented Massive Access: A Novel Strategy for Ultra-reliable Low-latency Massive Connectivity in Beyond 5G}
\title{Blind Goal-Oriented Massive Access\\ for Future Wireless Networks}

\author{Sajad Daei, 
\IEEEmembership{Member, IEEE} and Marios Kountouris, \IEEEmembership{Senior Member, IEEE}
\thanks{The authors are with the Communication Systems Department at EURECOM, Sophia-Antipolis, France, email: \{\texttt{sajad.daei, marios.kountouris\}@eurecom.fr}}}

\maketitle

\begin{abstract}
Emerging communication networks are envisioned to support massive wireless connectivity of heterogeneous devices with sporadic traffic and diverse requirements in terms of latency, reliability, and bandwidth. Providing multiple access to an increasing number of uncoordinated users and sharing the limited resources become essential in this context. In this work, we revisit the random access (RA) problem and exploit the continuous angular group sparsity feature of wireless channels to propose a novel RA strategy that provides low latency, high reliability, and massive access with limited bandwidth resources in an all-in-one package. To this end, we first design a reconstruction-free goal-oriented optimization problem, which only preserves the angular information required to identify the active devices. To solve this, we propose an alternating direction method of multipliers (ADMM) and derive closed-form expressions for each ADMM step. Then, we design a clustering algorithm that assigns the users in specific groups from which we can identify active stationary devices by their angles. For mobile devices, we propose an alternating minimization algorithm to recover their data and their channel gains simultaneously, which allows us to identify active mobile users. Simulation results show significant performance gains in terms of active user detection and false alarm probabilities as compared to state-of-the-art RA schemes, even with limited number of preambles. Moreover, unlike prior work, the performance of the proposed blind goal-oriented massive access does not depend on the number of devices.

\end{abstract}

\IEEEpeerreviewmaketitle

\section{Introduction}\label{section1}
\IEEEPARstart{U}{biquitous} wireless connectivity and its continuous evolution will increasingly play a critical role in people's everyday life in the upcoming years. The unprecedented growth of Internet of Thing (IoT) devices in beyond 5G (B5G) and 6G communication systems will create various new applications and services: smart cities and home automation by intelligent appliances, smart manufacturing in factories for providing informed decisions to the robotics, autonomous vehicles for smart transportation, health care monitoring for better care choices and remote surgery, delivering smart education for students, to name a few. Targeting the emergence of these applications, 5G and B5G/6G specifications have identified two inevitable use cases for machine-type user equipment (UE)s in IoT, namely ultra reliable, low latency communication (URLLC) and massive machine-type communication (mMTC). These two features will co-exist in IoT, enabling data transmission from myriads of UEs anywhere and anytime \cite{ho2019next,mahmood2020six}. 

Random access (RA) is a key yet challenging component of the communication process between UEs and base station (BS) in wireless networks, and particularly in emerging generations (5G and B5G/6G). The conventional RA strategy that has commonly been employed in human-type communications is grant-based (GB) RA \cite{hasan2013random}, which consists of four main stages. In the first stage, each active UE randomly selects one of the predefined preambles from a pool of orthogonal preamble sequences and sends it to the BS. In the second stage, the BS allocates resources to the activated preambles and sends an RA response as a grant for transmitting in the next stage. In the third stage, each UE that has received a response from the BS sends a connection request message in order to demand resources for data transmission. When there is no collision for the preambles, the BS sends a contention resolution message in the fourth stage to notify active UEs of the resources pending for data transmission. In case multiple UEs select the same preamble and resources, the BS detects the collision and does not respond to the affected UEs in the fourth stage; these UEs have to restart the RA process after waiting a random time. The number of available orthogonal preambles is directly proportional to the size of channel coherence block, which is limited. Therefore, under grant-based RA, in addition to the issue of significant signaling overhead, the number of active UEs that can be granted access to the network is limited by the number of preambles. Although several contention-based strategies have recently been proposed to reduce the collision probability (e.g., see \cite{bjornson2017random,bursalioglu2016rrh}), they have severe limitations and may not be used in mMTC scenarios for several reasons: (i) mMTC is expected to operate in crowded traffic scenarios due to a massive number of devices involved; (ii) MTC/IoT UEs are often battery-operated with limited power and bandwidth; (iii) device activity patterns are sporadic and at a certain time, only few devices are active, i.e, are in sleep mode most of the time and are sporadically activated to send data, and (iv) due to short data payloads and low latency requirements, random access and data transmission have to be performed together.

Recently, an alternative RA strategy, named grant-free (GF), has been proposed for mMTC in 5G communication systems \cite{kim2020two,liu2018massive,liu2018massive2,chen2021sparse,senel2018grant,haghighatshoar2018improved,ke2020compressive,djelouat2021joint,fengler2019massive,fengler2022pilot,xie2022massive,liu2018sparse,chen2019covariance,chen2018sparse,RA_daei,huang2018deep,ahn2021active,jang2021deep}. Unlike grant-based access, active UEs in grant-free access do not have to wait for a response/grant from the BS to send their data payloads; instead, data packets are transmitted without reserving channel resources in a time division
multiplexing (TDM) manner. Furthermore, unlike GB-RA where preambles are randomly selected at each time slot, preambles in GF-RA serve as the identifier (ID) of UEs during all time slots. This strategy leads to a significant reduction in the access delay of mMTC UEs. Exploiting orthogonal preambles in GF access results in the same aforementioned issues for GB access. A key limitation is that it is not practical or feasible to assign unique orthogonal preambles to a massive number of UEs during all time slots due to limited channel coherence time. For that, several works consider non-orthogonal preambles in GF (e.g., \cite{ke2020compressive,liu2018massive}). While using non-orthogonal preambles can reduce the collision probability, it also degrades the performance of active user detection (AUD) and channels estimation (CE). It is shown in \cite{ding2019analysis} that exploiting non-orthogonal preambles does not necessarily lead to higher access rate than its orthogonal counterpart. Moreover, using a large number of non-orthogonal preambles, which is a necessity for such type of methods, is not well suited for mMTC due to limited bandwidth requirements. There are also strategies combining grant-free and grant-based, known as semi-GF \cite{ding2019simple}. Semi-GF improves the performance compared to GF, while exhibiting lower signaling overhead and latency compared to GB. Nevertheless, semi-GF cannot support massive access and connectivity. 

Given the aforementioned limitations, the quest for a random access method that supports massive access without preamble collision while simultaneously ensuring not increasing the number of preambles and signaling overhead remains open. This is a timely problem of primary theoretical and practical relevance and is the main motivation of this work. This has not yet been explored in the literature of RA. We summarize below the latest RA works and their issues, categorizing them into three groups as follows:
\begin{itemize}
	\item Compressed sensing (CS). Several RA methods, such as \cite{ke2020compressive,lin2018estimation,xie2022massive,fengler2022pilot}, employ one of the well known algorithms in CS known as approximate message passing (AMP) to promote the discrete angular sparsity of MIMO channels. \cite{ke2020compressive} proposes to use an AMP variant that jointly performs both the tasks of active user detection and CE. When the number of preambles (i.e., time resources) is very large, AMP provides acceptable results, however it performs poorly when the number of preambles is low. Channel and noise distributions in AMP are not arbitrary and should follow specific predefined rules to work properly. A large pilot (training) overhead is required, making AMP resource wasteful. There is also a very recent work \cite{RA_daei} in this regard based on continuous CS that performs the tasks of AUD and CE; however, it remains computationally intractable when the number of devices exceeds a certain number (e.g., $K>30$) and is suitable for settings with limited number of users.   
	\item Covariance-based AUD \cite{covarianced-based2020,haghighatshoar2018improved,fengler2019massive,haghighatshoar2018new,djelouat2021joint,chen2021phase}. This type of RA performs only AUD and not CE. It is based on calculating the covariance of measurement matrix. Based on the covariance measurement matrix, a maximum likelihood (ML) method is then used to detect the active devices. Again, channel and noise distributions should follow specific rules and cannot be arbitrary. The complexity increases with the number of devices. It only considers the sparsity of MTC devices and not the inherent features of channels, e.g., the continuous angular sparsity. Moreover, a huge number of antennas is required for these methods to work properly.
	\item Deep learning (DL) \cite{ahn2021active,huang2018deep,jang2021deep}. A number of works use deep neural networks to approximate the mapping between the received and transmitted signals for joint AUD and CE. This type of works usually requires a prohibitively large amount of training data, which indeed takes a very long time to collect and label, and seems to be hard to obtain in practice due to the sporadic traffic of MTC devices. Apart from that, when the system parameters change for example in subsequent coherence intervals, DL spends severe amount of resources and computational complexity to retrain the system, which is extremely burdensome and does not scale for fulfilling the stringent latency requirements of B5G systems.
\end{itemize}
%\subsection{Prior arts and their issues}
%There are many RA methods in the literature such as . In what follows, we divide these works into three categories and provide a summary of these results:\cite{ke2020compressive,djelouat2021joint,fengler2019massive,fengler2022pilot,xie2022massive,liu2018sparse,liu2018massive,bjornson2017random,chen2021sparse,chen2019covariance,chen2018sparse,RA_daei,huang2018deep,ahn2021active,jang2021deep}

\subsection{Our Approach and Contributions}
A unique, novel feature of our proposed random access strategy, coined blind goal-oriented detection (BGOD), is that it incorporates all characteristics of B5G requirements in all-in-one package. Our method builds upon the assumption that MIMO channels exhibit angular continuous \emph{group sparsity} features. This implies that only few number of components in the angular domain contribute to the channel of each UE and the angles of arrival (AoAs) corresponding to each UE lie alongside each other in a group, as shown in Figure \ref{fig.channel}(a) (the interested reader is referred to \cite{lin2018estimation} for detailed description of this feature). We first design a general optimization problem that promotes the angular group sparsity features of the channels corresponding to all UEs. This optimization problem is highly challenging and costly in terms of computational complexity when the number of UEs is very large (as is the case for massive access in B5G/6G wireless networks), and actually computationally intractable or infeasible in practice. Our first goal is to design a simple reconstruction-free framework that only identifies all active UEs, bypassing the difficult task of data recovery (DR) and CE for massive UEs. To this end, we design a novel goal-oriented optimization problem that does not keep information on the UEs' messages and complex channel coefficients, but it provides instead a way to obtain a goal-oriented continuous function whose maximum value interestingly reveals the angle information of all active UEs. We then design a clustering algorithm to place these angles into several groups and find the ID of active UEs via their angles. This is indeed the case when UEs are stationary (e.g., smart metering MTC devices) and the angles of each UE, e.g., the line of sight (LoS) angle, are known to the BS prior to the RA process. For the case where UEs are moving or their angles are not known to the BS in advance, we design an alternating minimization (AM) method to recover both the messages and the complex channel coefficients corresponding to active UEs. The ID of each active user is contained in its data payload. By the proposed strategy, the ID and data payloads of active UEs and the channels can be fully obtained blindly.

The novel features and the main contributions of our work are summarized as:  
\begin{enumerate}
	\item \textbf{Massive connectivity.} Interestingly, the proposed goal-oriented optimization problem does not depend on the number of total users, whereas the complexity of our clustering scheme does depend on the number of active UEs. Therefore, a massive number of devices with sporadic traffic can access to the network. Actually, the higher the number of antennas at the BS, the more active UEs can access the network. To the best of our knowledge, this is the first work where the RA strategy is independent of the number of devices.
	\item \textbf{High reliability.} Employing massive number of antennas at the BS and orthogonal preambles at the devices leads to precise detection of the angles and facilitates reliable data recovery. The lower the maximum length of groups (known as angular spread), the fewer observations are required at the BS for exact recovery of angles. Interestingly, devices that selected the same orthogonal preambles (due to limitation in channel coherence time) can still be distinguished by their different angles. The possibility that two active users exist at the same time with the same preamble and with the same continuous angles is very low and can be ignored.
	\item \textbf{Low latency.} We design an alternating direction method of multipliers (ADMM) in order to directly and rapidly obtain the goal-oriented continuous function by a few measurements obtained at the BS. The proposed algorithm is very fast and may fulfil the low latency requirement of B5G/6G systems. Moreover, the proposed algorithm operates in a blind way, i.e, the BS does not need to use pilots for channel estimation and does not need to coordinate with users in advance. This makes our RA strategy very fast without any extra access latency. Since there are no collisions of active users in our method, there is no need for any contention resolution step. Therefore, as shown in Figure \ref{fig.channel}(c), the proposed RA strategy is performed in only one step and with very low access delay.
	\item \textbf{Limited resources.} The goal-oriented feature introduced in this work opens up the possibility of spending as few orthogonal preambles as possible in a blind way, which results in significant resource saving. There is no need to use known pilots for channel estimation and the devices can transmit only their data. This feature is important for MTC/IoT devices, which are battery-operated and operate using limited bandwidth and low power. 
\end{enumerate}
\begin{figure}[t]
\hspace{-0cm}
	%{\includegraphics[trim={1.9cm .6cm 0cm 2.2cm}]{channel.tex}}
\begin{subfigmatrix}{2}
\subfigure[]{\includegraphics[scale=1,trim={0cm 0cm 0cm 3cm}]{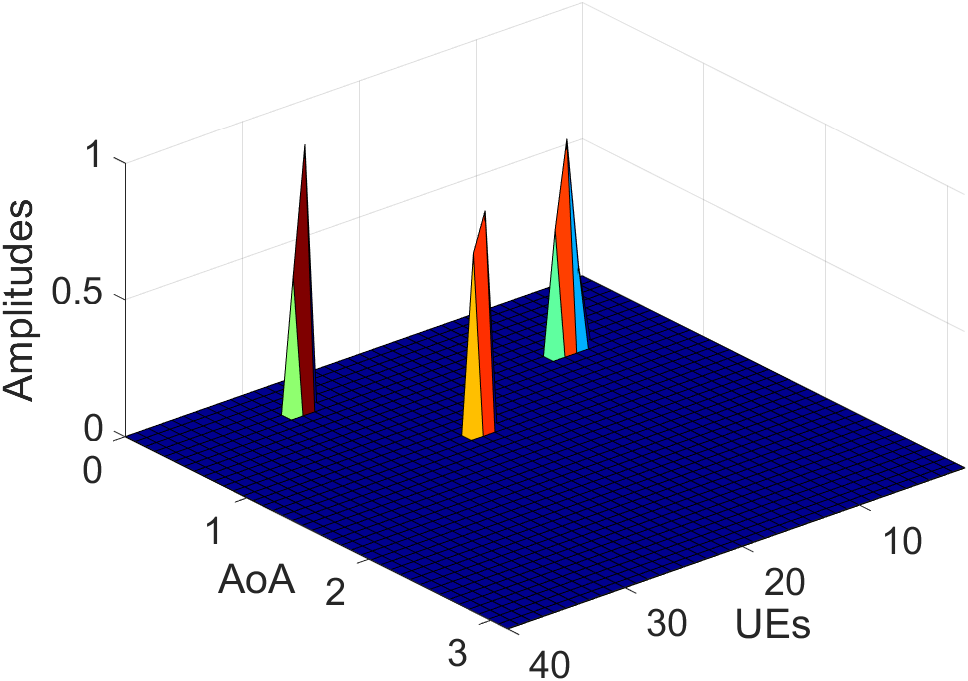}}
\subfigure[]{\includegraphics[scale=1,trim={0cm 0cm 0cm 0cm}]{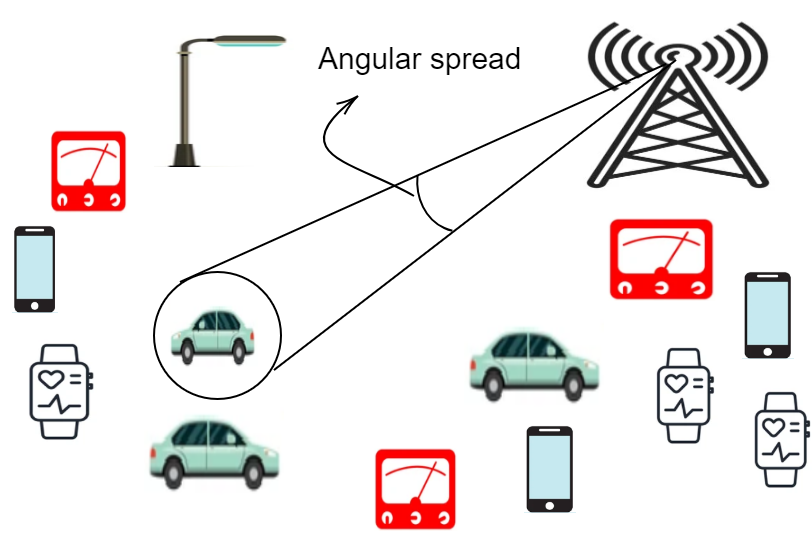}}
\subfigure[]{\includegraphics[width=9cm,trim={0cm 0cm 0cm 0cm}]{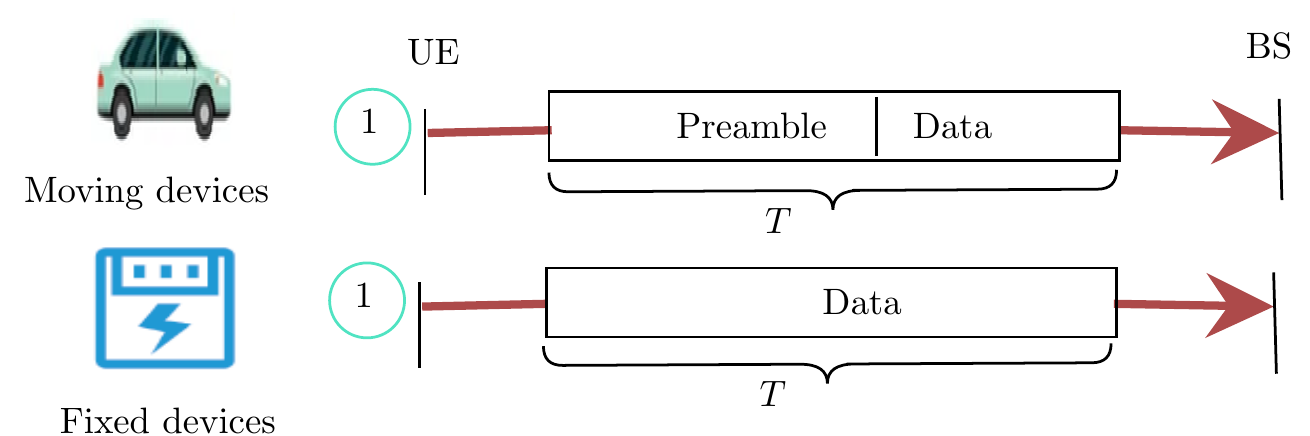}}
\end{subfigmatrix}
\caption{(a) Continuous angular group sparsity in the uplink channel. Only $K_a=3$ UEs are active and the maximum number of physical groups are considered to be three ($L_{\max}:=\max_{k}L_k=3$). (b) The angular spread of each UE with respect to the BS. (c) Proposed blind massive access scheme. Bottom image: The BS has full access to the angle information of fixed devices. Top image: UEs are moving and send random preambles to the BS, which accounts for their identification.}\label{fig.channel}
\end{figure}

\subsection{Organization} 
The remainder of this paper is structured as follows. Section \ref{sec.model} introduces the considered wireless system model with its specific features. In Section \ref{sec.proposed}, we present the proposed blind goal-oriented RA approach, which consists of four steps: goal-oriented optimization \ref{sec.goal-oriented}, ADMM \ref{sec.admm}, active UE identification \ref{sec.active_user}, and message recovery \ref{sec.data_rec}. In Section \ref{sec.simulations} we assess the performance of our algorithm and compare it with state-of-the-art RA schemes using numerical experiments. Finally, Section \ref{sec.conclusion} concludes the paper.
\subsection{Notations}
% \textit{Notations}: 
We use boldface lower and uppercase letters for vectors and matrices, respectively. The $i$-th element of a vector $\bm{x}$ and the $(i,l)$ element of a matrix $\bm{X}$ are respectively denoted by $x_i$ and $X_{(i,l)}$. The notation $j$ is employed to represent the imaginary unit. The real and imaginary parts of a complex-valued matrix $\bm{A}=\bm{A}^R+j\bm{A}^I\in\mathbb{C}^{n_1\times n_2}$ are shown by $\bm{A}^{R}$ and $\bm{A}^I$, respectively, and $\overline{\bm{A}}:=\bm{A}^R-j\bm{A}^I$ denotes the conjugate. For vector $\bm{x}\in\mathbb{C}^n$ and matrix $\bm{X}\in\mathbb{C}^{n_1\times n_2}$, the $\ell_2$ norm and Frobenius norm are defined as $\|\bm{x}\|_2:=({\sum_{i=1}^n|x(i)|^2})^{\tfrac{1}{2}}$ $\|\bm{X}\|_{F}:=\sqrt{\sum_{i=1}^{n_1}\sum_{j=1}^{n_2}|X(i,j)|^2}$, respectively. $\bm{X}\succeq \bm{0}$ means that $\bm{X}$ is a positive semidefinite matrix. $\mathcal{P}_{\Omega}(\cdot)$ is a operator transforming an arbitrary matrix to a reduced matrix with rows indexed by $\Omega$. $\langle \bm{A},\bm{B}\rangle:=\sum_{i=1}^{n_1}\sum_{l=1}^{n_2}A_{i,l}\overline{B}_{i,l}$ is the inner product of two complex-valued matrices $\bm{A}\in\mathbb{C}^{n_1\times n_2}$ and $\bm{B}\in\mathbb{C}^{n_1\times n_2}$. The Toeplitz matrix $\mathcal{T}(\bm{v})$ is defined as
\begin{align}\label{eq.toeplitz_mat}
	\mathcal{T}(\bm{v})=\begin{bmatrix}
		v_1&v_2&\hdots&v_{N}\\
		\overline{v}_2&v_1&\hdots&v_{N-1}\\
		\vdots&\vdots&\ddots&\vdots\\
		\overline{v}_N&\overline{v}_{N-1}&\hdots&v_1
	\end{bmatrix}
\end{align}
where the $(i,l)$-th element is given by $\mathcal{T}(\bm{v})_{(i,l)}=\left\{\begin{array}{cc}
	v_{i-l+1}&i\ge l\\
	\overline{v}_{l-i+1}& i<l
\end{array}\right\}$. $\bm{1}_{\Omega}$ is a vector of size $\mathbb{R}^N$ which has $1$s on the indices corresponding to the set $\Omega$ and zero elsewhere. ${\rm diag}(\bm{x})$ is a diagonal matrix whose main diagonal has elements of $\bm{x}\in\mathbb{R}^N$. $\bm{e}_1\in\mathbb{R}^N$ is the canonical vector whose first element is $1$ and zero elsewhere, i.e., $e_1(1)=1,~ e_1(i)=0, i=2,..., N$. $\bm{x}\odot\bm{y} \in\mathbb{C}^{N}$ is the element-wise operation of two vectors $\bm{x}\in\mathbb{C}^N$ and $\bm{y}\in\mathbb{C}^N$ and its $i$-th element is given by $x_i y_i$. The element-wise inequality for two vectors $\bm{x}\in\mathbb{R}^N$ and $\bm{y}\in\mathbb{R}^N$ is represented by $\bm{x}\ge \bm{y}$ which means $x_i\ge y_i, i=1,..., N$. $\mathds{E}$ and $\mathds{P}$ denote expectation and probability, respectively.

\section{System Model}\label{sec.model}
%\begin{figure}
%	
%\caption{Channel model}
%\end{figure}
We consider a wireless system in which a BS equipped with uniform linear array (ULA) of $N$ antennas is serving a large number $K$ of single-antenna UEs out of which $K_a$ are active denoted by the set $\mathcal{S}_{\rm AU}$. We assume a block fading channel, which remains constant over the coherence time $T$ and varies smoothly between adjacent coherence blocks. The channel vector in the frequency domain from $k$-th UE to the BS can be represented by \cite{you2015pilot,bajwa2010compressed},\cite[Equ. 7]{zhang2017blind}:
\begin{align}\label{eq.chann}
	\bm{h}_k=\int_{0}^{\pi}\alpha_k(\theta)\bm{a}(\theta){\rm d}\theta
\end{align}
where $\alpha_k(\theta)$ is the channel gain of user $k$ corresponding to the direction $\theta$ and $\bm{a}(\theta)$ is the array response vector of BS defined by
\begin{align}\label{eq.atoms}
	\bm{a}(\theta)=\tfrac{1}{\sqrt{N}}[1, {\rm e}^{-j2\pi \frac{d}{\lambda} \cos(\theta)},...,{\rm e}^{-j2\pi \frac{d}{\lambda} (N-1)\cos(\theta)} ]^T
\end{align}
where $\lambda$ and $d$ are the carrier wavelength and antenna spacing, respectively.
We further assume that there is limited local scattering around the BS and the channel gains $\alpha_k(\theta)$ are constrained to lie in a small region $(\theta_k^{\min},\theta_k^{\max})$ known as angular spread and composed of $L_k\ll N$ spatial angles of arrival (AoA) \cite{ke2020compressive,djelouat2021joint,ma2018sparse} (see Figures (a) and (b) of \ref{fig.channel}). Thus, \eqref{eq.chann} can be rewritten as
\begin{align}\label{eq.channel}
	\bm{h}_k=\sum_{l=1}^{L_k}\alpha_l^k\bm{a}(\theta_l^k)=:\bm{A}_k\bm{\alpha}^k\in\mathbb{C}^{N\times 1}
\end{align}
where $\alpha_l^k$ accounts for the gain of the $l$-th path, $\theta_l^k$ is the AoA of the $l$-th path for the $k$-th user, $\bm{\alpha}^k:=[\alpha_1^k,..., \alpha_{L_k}^k]^T$ and $\bm{A}_k:=[\bm{a}_r(\theta_1^k),..., \bm{a}_r(\theta_{L_k}^k)]\in\mathbb{C}^{N\times L_k}$. We consider the case where the BS has partial observations, i.e., only signals received by $M$ out of $N$ antennas, indexed by $\Omega \subseteq \{1,..., N\}$ ($|\Omega|=M$), are observed and the rest can be used for other purposes, e.g., serving the UEs at the other side of BS. By considering $\bm{\phi}_k\in\mathbb{C}^T$ as the information transmitted by $k$-th user including both preamble and data, the received signal after $T$ time slots at the BS is given by \cite{liu2018massive,chen2018sparse}:
\begin{eqnarray}\label{eq.observed}
\bm{Y}_{\Omega}&=&\mathcal{P}_{\Omega}(\bm{Y}) =\mathcal{P}_{\Omega}\left(\sum_{k\in\mathcal{S}_{\rm AU}}\bm{h}_{k}\bm{\phi}_k^H\right)+\bm{E}\nonumber\\
	&\coloneqq&\sum_{k\in\mathcal{S}_{\rm AU}}\mathcal{P}_{\Omega}(\bm{X}_k)+\bm{E}\in \mathbb{C}^{M\times T}
\end{eqnarray}  
where $\bm{E}\in\mathbb{C}^{M\times T}$ is the additive noise matrix, which has arbitrary distribution with $\|\bm{E}\|_F\le \eta$ and $\mathcal{S}_{\rm AU}\subseteq \{1,..., K\}$ is the set of active users. Inspired by \cite{chandrasekaran2012convex,candes2014towards,tang2013compressed}, each matrix $\bm{X}_k:=\sum_{l=1}^{L_k}\alpha_l^k\bm{a}(\theta_l^k)\bm{\phi}_k^{H}$ is a superposition of $L_k$ building blocks (referred to as \textit{atoms}) of the form $\bm{a}(\theta_l^k)\bm{\phi}_k^{H}$. Define the set of building blocks as an atomic set ${ \mathcal{A}_k=\{\bm{a}(\theta)\bm{\phi}_k^H, \theta\in (0,\pi)\}}$. Each $\bm{X}_k$ is composed of a sparse set of atoms in $\mathcal{A}$.

\section{Blind Massive Random Access}\label{sec.proposed}
In this section, we present the proposed massive blind random access strategy, which consists of four main stages:
\begin{enumerate}
    \item goal-oriented optimization
    \item ADMM 
    \item identification of active UEs
    \item joint data recovery (DR) and channel estimation (CE). 
\end{enumerate}
It is worth mentioning that the first three steps are sufficient to identify active UEs distinguishable by their AoAs, e.g., stationary MTC UEs. The above four steps are summarized in Algorithm \ref{algorithm.admm}.

\subsection{Goal-oriented Optimization} \label{sec.goal-oriented} In \eqref{eq.observed}, we have a system of equations with a very large number of unknowns (i.e., $K N T$) and only $M T$ observations at the BS. Leveraging the degrees of freedom of the problem in \eqref{eq.observed}, which is $\sum_{k\in\mathcal{S}_{\rm AU}}(L_k+T)$, motivates us to use a general optimization problem that encourages the features corresponding to all UEs simultaneously similar to what $\ell_0$ functional offers to encourage sparsity feature in conventional CS. Capitalizing on the results from continuous compressed sensing (CS) (see e.g., \cite{sayyari2020blind,bayat2020separating,daei2019living,candes2014towards,tang2013compressed,fernandez2016super}), such general optimization framework is as follows:
\begin{align}\label{prob.atomic_l0}
	&\min_{\substack{\bm{Z}_k\in\mathbb{C}^{N\times T}\\ k=1,..., K}} \sum_{k=1}^K \|{\bm{Z}_k}\|_{\mathcal{A}_k,0} ~s.t. \|\bm{Y}-\sum_{k=1}^K\mathcal{P}_{\Omega}(\bm{Z}_k)\|_F\le \eta
\end{align} 
where $\|\bm{Z}_k\|_{\mathcal{A},0} \coloneqq \inf\big\{L_k: \bm{Z}_k=\sum_{l=1}^{L_k}c_{l}^k\bm{a}({\theta}_l^k)\bm{\phi}_k^H, c_l^k>0 , \bm{a}(\theta_l^k)\bm{\phi}_k^H\in\mathcal{A}_k\big\}
$ is the atomic $\ell_0$ function that computes the least number of atoms needed to describe $\bm{Z}_k$ by the atoms in the atomic set $\mathcal{A}_k$ and $\sum_{k=1}^K \|{\bm{Z}_k}\|_{\mathcal{A}_k,0}$ is a function that promotes $\sum_{k=1}^K L_k$. We reformulate \eqref{prob.atomic_l0} into an equivalent form called LASSO-type and add a regularization term to ensure consistency of the measurements given by:
\begin{eqnarray}\label{prob.atomic_l0_lasso}
	&&\min_{\substack{\bm{Z}_k\in\mathbb{C}^{N\times T}\\ k=1,..., K}} \sum_{k=1}^K \|{\bm{Z}_k}\|_{\mathcal{A}_k,0}+\frac{\gamma}{2}\|\bm{Y}-\bm{Y}^{\star}\|_F^2 \nonumber\\
	&&~\textrm{s.t.}~~ \bm{Y}^{\star}=\sum_{k=1}^K\mathcal{P}_{\Omega}(\bm{Z}_k)
\end{eqnarray} 
where $\gamma>0$ is a regularization parameter that makes a balance between the noise energy and the angular sparsity feature with a role similar to $\eta>0$ in \eqref{prob.atomic_l0}. Nevertheless, the optimization problem \eqref{prob.atomic_l0_lasso} is NP-hard and intractable in general. It is therefore beneficial to work with the nearest tractable convex optimization problem whose objective function is a convex relaxation of that in \eqref{prob.atomic_l0_lasso} and is stated as follows: 
\begin{eqnarray}\label{prob.atomic_l1_lasso}
	&&\min_{\substack{\bm{Z}_k\in\mathbb{C}^{N\times T}\\ k=1,..., K\\\bm{Y}^{\star}\in\mathbb{C}^{M\times T}}} \sum_{k=1}^K \|{\bm{Z}_k}\|_{\mathcal{A}_k}+\frac{\gamma}{2}\|\bm{Y}-\bm{Y}^{\star}\|_F^2 \nonumber\\
	&&~\textrm{s.t.}~~ \bm{Y}^{\star}=\sum_{k=1}^K\mathcal{P}_{\Omega}(\bm{Z}_k)
\end{eqnarray}
where the atomic norm $\|\cdot\|_{\mathcal{A}_k}$ is the best convex surrogate for the number of atoms composing $\bm{Z}_k$ (i.e., $\|\cdot\|_{\mathcal{A}_k,0}$) and is defined as the minimum of the $\ell_1$ norm of the coefficients forming $\bm{Z}_k$:
\begin{align}
	&\|\bm{Z}_k\|_{\mathcal{A}_k}\coloneqq\inf\{t>0: \bm{Z}_k\in t{\rm conv}(\mathcal{A}_k)\}=\nonumber\\
	&\inf\{\sum_{l=1}^{L_k}c_l^k: \bm{Z}_k=\sum_{l=1}^{L_k}c_{l}^k\bm{a}({\theta}_l^k)\bm{\phi}_k^H, c_l^k>0, \bm{a}(\theta_l^k)\bm{\phi}_k^H\in\mathcal{A}_k\}
\end{align}
where ${\rm conv(\mathcal{A})}$ is the convex hull of $\mathcal{A}$. Despite convexity of the objective function in \eqref{prob.atomic_l0_lasso}, it remains computationally intractable due to the continuous nature of the angles $\theta_k$s. 
The objective of this work is to identify the AoAs corresponding to active users. The following theorem provides a novel reconstruction-free optimization problem, which contains only the information of AoAs corresponding to active users and ignores the message information and the channel coefficients corresponding to massive number of UEs. In fact, this theorem provides a goal-oriented framework of solving \eqref{prob.atomic_l1_lasso}, in which the \emph{goal} is only restricted to active user detection, i.e., finding the AoAs corresponding to the active users. Before stating this theorem, we first need to define the minimal wrap-around distance (also called \textit{minimum separation}) between angles as
\begin{align}
	\Delta \coloneqq \min_{k=1,..., K}\min_{i\neq q}|\cos(\theta_i^k)-\cos(\theta_q^k)|
\end{align} 
where the absolute value only in the latter relation is evaluated over the unit circle, e.g., $|0.01-0.09|=0.02$.

\begin{thm}\label{thm.main}
Let $c_1 \coloneqq \frac{\max_{k=1,\ldots, K}\|\bm{\phi}_k\|_2}{\sqrt{N}}$. Suppose that the continuous AoAs of active users do not have any intersections with each other, which indeed implies that at each given time, two users with the same continuous AoAs should not be active. Consider $\bm{V}\in\mathbb{C}^{M\times T}$ as the dual matrix variable corresponding to the primal variable $\bm{Y}^{\star}$ in the following optimization problem:
	\begin{eqnarray}\label{prob.goal_optimization}
		&&\min_{\substack{\bm{v}\in\mathbb{C}^{N},\bm{Z}\in \mathbb{C}^{N\times T}\\
				\bm{Y}^{\star}\in\mathbb{C}^{{M\times T}}, \bm{W}\in\mathbb{C}^{T\times T}}} {\rm Re}(v_1)+{\rm Re}({\rm tr}(\bm{W}))+\frac{\gamma}{2}\|\bm{Y}-\bm{Y}^{\star}\|_F^2\nonumber\\
		&&{\textrm{s.t.}}~~\begin{bmatrix}
			\mathcal{T}(\bm{v})&\bm{Z}\\
			\bm{Z}^H&\bm{W}
		\end{bmatrix}\succeq \bm{0}~~, \bm{Y}^{\star}=-2c_1\mathcal{P}_{\Omega}(\bm{Z}).
	\end{eqnarray}	
Then, the AoAs corresponding to the active users are uniquely identified provided that $\Delta> \frac{1}{N}$ by finding the angles that maximize the $\ell_2$ norm of the goal-oriented dual polynomial $\bm{q}_G(\theta)$ as follows:
	\begin{align}\label{eq.angle_find}
		\widehat{\theta}_l^k=\mathop{\arg\max}_{\theta\in (0,\pi)} \|\bm{q}_G(\theta)\|_2, \ l=1,\ldots, L_k, k\in\mathcal{S}_{\rm AU}
	\end{align}
where $\bm{q}_G(\theta)=(\mathcal{P}_{\Omega}^{{\rm Adj}}(\bm{V}))^H\bm{a}(\theta)$.
\end{thm}
\begin{proof}
See Appendix \ref{proof.thm_main}.
\end{proof}

\textit{Discussion}: The above theorem shows that all necessary information required for identifying active users is contained in the goal-oriented dual polynomial $\bm{q}_G(\theta)$. Figure \ref{fig.dualpol} shows a typical example of the $\ell_2$ norm of the goal-oriented function $\|\bm{q}_G(\theta)\|_2$. The active devices can be identified by finding the angles that maximizes $\|\bm{q}_G(\theta)\|_2$. The proposed optimization problem \eqref{prob.goal_optimization} is independent of the number of devices $K$ and can be solved for instance using CVX \cite{cvx} (CVX calls for SDPT3 solvers). To obtain $\|\bm{q}_G(\theta)\|_2$, we need the matrix dual solution $\bm{V}$. Most of the CVX solvers return the dual variables along with the primal variables. However, solving this optimization problem is highly challenging, in particular for massive MIMO communications where the BS is equipped with a very large number of antennas $N$. Thus, the SDPT3 solver of CVX cannot reach to a solution. In order to solve this this issue, we design a fast ADMM method (see Section \ref{sec.admm}), which significantly reduces the computational complexity and is tailored to problems encountered in massive access.
\begin{rem}(Relation between detection accuracy and $N$)
The higher the number of BS antennas $N$, the easier the minimum separation condition $\Delta>\frac{1}{N}$ gets, thus the peaks of $\|\bm{q}_G(\theta)\|_2$ can be identified more clearly and $\|\bm{q}_G(\theta)\|_2$ provides more information about active devices.
\end{rem}
\begin{rem}(Relation between detection accuracy and $T$)
The more time resources are employed at the MTC devices, the more information can be provided by $\|\bm{q}_G(\theta)\|_2$, resulting in easier identification of active devices.
\end{rem}
\begin{rem}(Detection accuracy versus $K$ and $K_a$)
The accuracy and the complexity of the proposed framework do not depend on the number of devices $K$. For moving UEs (only), the higher the number of active devices, the more difficult distinguishing which angles correspond to them is, and as a result, the detection becomes harder.	
\end{rem}
\begin{figure}[!t]
	\hspace{-.5cm}
	\includegraphics[scale=.42]{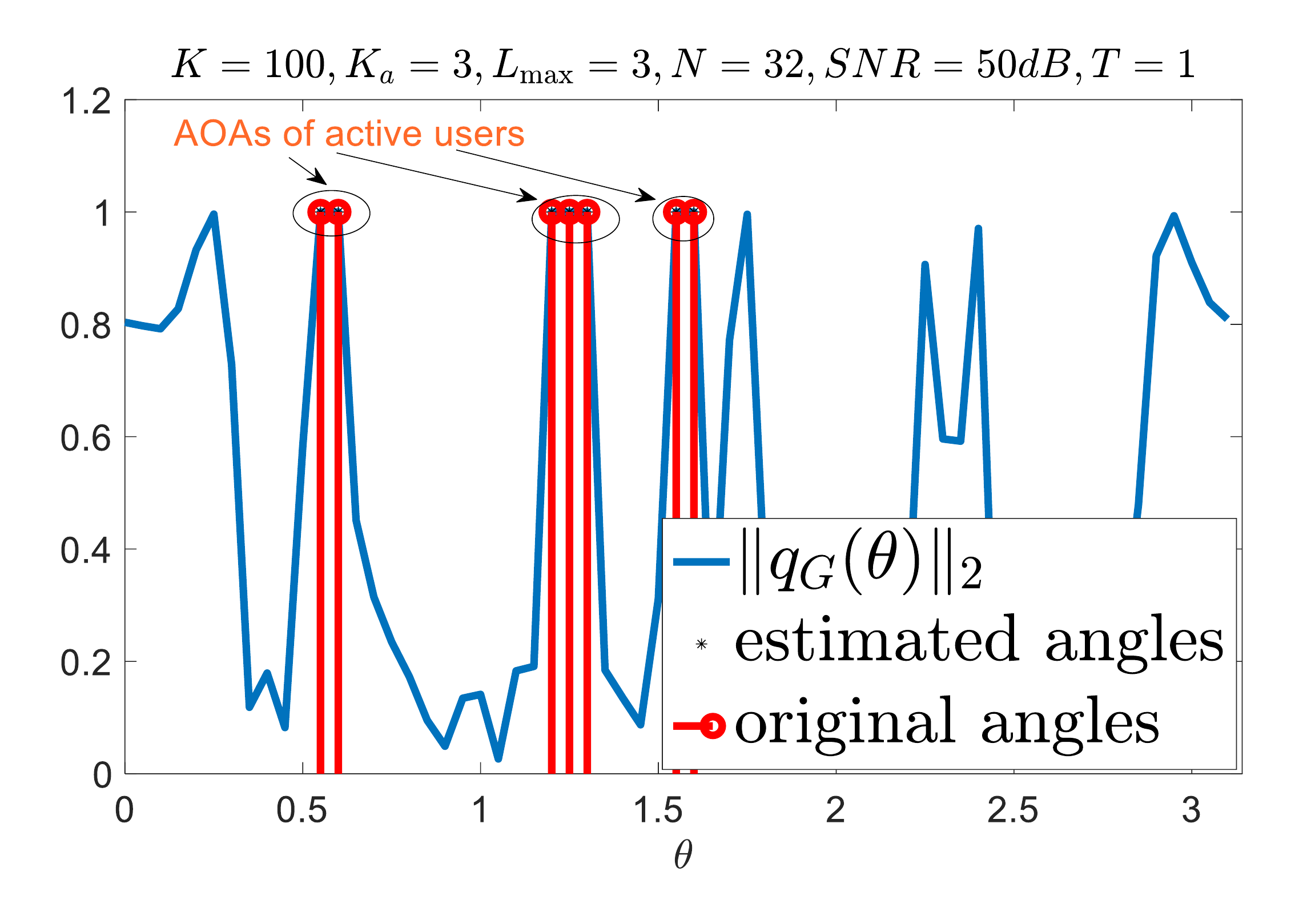}
	\caption{The $\ell_2$ norm of the goal-oriented dual polynomial function $\|\bm{q}_{G}(\theta)\|_2$. The angles for which $\|\bm{q}_{G}(\theta)\|_2$ achieves its maximum determine the angles of active UEs.}\label{fig.dualpol}
\end{figure}

\subsection{Proposed ADMM}\label{sec.admm}
Solving the optimization problem \eqref{prob.goal_optimization} with SDPT3 is prohibitive when $N$ and $T$ are very large. For that, we propose here an ADMM algorithm that can be executed in a significantly faster time. The general idea is to first form an augmented Lagrangian function for the optimization problem and then to split it into a sum of separable functions. In fact, each step of ADMM involves finding a local minimum on a variable \cite{boyd2011distributed}. The steps are repeated until some stopping criteria is satisfied. To apply this method, we first find an augmented version of \eqref{prob.goal_optimization} by incorporating an intermediate variable $\bm{\Psi}$ into \eqref{prob.goal_optimization} in order to decouple the semidefinite constraint from the affine constraints. This leads to the following equivalent problem:
\begin{align}\label{prob.admm1}
	&\min_{\substack{\bm{v}\in\mathbb{C}^{N},\bm{Z}\in \mathbb{C}^{N\times T}\\
			\bm{W}\in\mathbb{C}^{T\times T}\\
			\bm{\Psi}\in \mathbb{C}^{N+T\times N+T}}} {\rm Re}(v_1)+{\rm Re}({\rm tr}(\bm{W}))+\frac{\gamma}{2}\|\bm{Y}+2c_1\mathcal{P}_{\Omega}(\bm{Z})\|_F^2\nonumber\\
	&{\rm s.t.}~~\begin{bmatrix}
		\mathcal{T}(\bm{v})&\bm{Z}\\
		\bm{Z}^H&\bm{W}
	\end{bmatrix}=\begin{bmatrix}
		\bm{\Psi}_0&\bm{\Psi}_1\\
		\bm{\Psi}_1^H&\overline{\bm{\Psi}}
	\end{bmatrix} \coloneqq \bm{\Psi}\\
	&\bm{\Psi}\succeq \bm{0}.
\end{align}  
The augmented Lagrangian function for \eqref{prob.admm1} is given as
\begin{eqnarray}\label{eq.L_rho}
	&&\mathcal{L}_{\rho}(\bm{v},\bm{W},\bm{Z},\bm{\Lambda},\bm{\Psi}) = {\rm Re}(v_1)+{\rm Re}({\rm tr}(\bm{W}))\nonumber\\
	&&+\frac{\gamma}{2}\|\bm{Y}+2c_1\mathcal{P}_{\Omega}(\bm{Z})\|_F^2\nonumber\\
	&&+{\rm Re}\Bigg\langle \begin{bmatrix}
		\bm{\Lambda}_0&\bm{\Lambda}_1\\
		\bm{\Lambda}_1 & \overline{\bm{\Lambda}}
	\end{bmatrix}, \begin{bmatrix}
		\bm{\Psi}_0&\bm{\Psi}_1\\
		\bm{\Psi}_1^H&\overline{\bm{\Psi}}
	\end{bmatrix}- \begin{bmatrix}
		\mathcal{T}(\bm{v})&\bm{Z}\\
		\bm{Z}^H&\bm{W}
	\end{bmatrix} \Bigg \rangle\nonumber\\
	&&+\frac{\rho}{2}\Bigg\|\begin{bmatrix}
		\bm{\Psi}_0&\bm{\Psi}_1\\
		\bm{\Psi}_1^H&\overline{\bm{\Psi}}
	\end{bmatrix}-\begin{bmatrix}
		\mathcal{T}(\bm{v})&\bm{Z}\\
		\bm{Z}^H&\bm{W}
	\end{bmatrix}\Bigg\|_F^2
\end{eqnarray} 
where $\rho>0$ is a regularization parameter and $\bm{\Lambda}=\begin{bmatrix}
	\bm{\Lambda}_0&\bm{\Lambda}_1\\
	\bm{\Lambda}_1 & \overline{\bm{\Lambda}}
\end{bmatrix}$ is the Lagrangian dual multiplier.  Denoting the $t$-th step of $\bm{\Psi}$ and $\bm{\Lambda}$ by $\bm{\Psi}^{t}$ and $\bm{\Lambda}^{t}$, ADMM successively updates the following steps to minimize the Lagrangian function $\mathcal{L}_{\rho}$:
\begin{align}
	&(\bm{v}^{t+1}, \bm{W}^{t+1}, \bm{Z}^{t+1})=\mathop{\arg\min}_{\substack{\bm{v}\in\mathbb{C}^{N}, \bm{W}\in\mathbb{C}^{T\times T},\\ \bm{Z}\in\mathbb{C}^{N\times T}}}\mathcal{L}_{\rho}(\bm{v},\bm{W},\bm{Z},\bm{\Lambda}^{t},\bm{\Psi}^{t})\\
	&\bm{\Psi}^{t+1}=\mathop{\arg\min}_{\bm{\Psi}\succeq \bm{0}}\mathcal{L}_{\rho}(\bm{v}^t,\bm{W}^t,\bm{Z}^t,\bm{\Lambda}^{t},\bm{\Psi})\\
	&\bm{\Lambda}^{t+1}=\bm{\Lambda}^{t}+\rho \Bigg(  \begin{bmatrix}
		\bm{\Psi}_0^{t+1}&\bm{\Psi}_1^{t+1}\\
		(\bm{\Psi}_1^{t+1})^H&\overline{\bm{\Psi}}^{t+1}
	\end{bmatrix}-\begin{bmatrix}
		\mathcal{T}(\bm{v}^{t+1})&\bm{Z}^{t+1}\\
		(\bm{Z}^{t+1})^H&\bm{W}^{t+1}
	\end{bmatrix} \Bigg),\label{Lambda_update}
\end{align}
where $t$ indicates the iteration number. To make the method practical, we need to obtain efficient implementations for all of the steps. Obviously, the Lagrangian function $L_{\rho}$ is convex with respect to $\bm{Z}$ and $\bm{W}$. Hence, we can obtain closed-form update solutions for $\bm{Z}$ and $\bm{W}$ by setting the partial derivative of $\mathcal{L}_{\rho}$ equal to zero as follows:
\begin{align}
	&\frac{\partial \mathcal{L}_{\rho}}{\partial \bm{Z}}=2c_1\gamma\mathcal{P}_{\Omega}^{{\rm Adj}}(\bm{Y}+2c_1 \mathcal{P}_{\Omega}(\bm{Z}))-2 \bm{\Lambda}_1-2\rho (\bm{\Psi}_1-\bm{Z})=\bm{0}\label{diff_Lrho_Z}\\
	&\frac{\partial \mathcal{L}_{\rho}}{\partial \bm{W}}=\bm{I}-\overline{\bm{\Lambda}}-\rho(\overline{\bm{\Psi}}-\bm{W})=\bm{0}
\end{align}
which lead to closed-form relations
\begin{align}\label{update_Z&W}
	&\bm{Z}^{t+1} = {\scriptstyle \bigg(4c_1^2\gamma \bm{P}_{\Omega}^{{\rm Adj}}\bm{P}_{\Omega}+2\rho \bm{I}_{N}\bigg)^{-1}\bigg(2\bm{\Lambda}_1^{t}+2\rho\bm{\Psi}_1^{t}-2\gamma c_1 \mathcal{P}^{{\rm Adj}}_{\Omega}\bm{Y}\bigg)}\\
	&\bm{W}^{t+1} = \overline{\bm{\Psi}}^{t}+\frac{\overline{\bm{\Lambda}}^{t}-\bm{I}}{\rho}
\end{align} 
where $\bm{P}_{\Omega}$ and $\bm{P}_{\Omega}^{{\rm Adj}}$ are matrix forms of the linear operators $\mathcal{P}_{\Omega}$ and $\mathcal{P}^{{\rm Adj}}_{\Omega}$ which by using MATLAB notations are obtained as
\begin{align}
	\bm{P}_{\Omega}=\bm{D}(\Omega,:),~\bm{P}^{{\rm Adj}}_{\Omega}=\bm{D}(:,\Omega)
\end{align} 
in which $\bm{D}={\rm diag}(\bm{1}_{\Omega})\in\mathbb{R}^{N\times N}$. The closed-form expression for $\bm{v}$ is provided in the following lemma:
\begin{lem}\label{lem.v_cal}
Let $\bm{g}\in\mathbb{R}^N$ be a vector whose elements are $g(1)=\frac{1}{N}$ and $g(k)=\frac{1}{2(N-k+1)},~ k\neq1$. Define $\bm{C}\in\mathbb{R}^{N\times N}$ as a matrix composed of $-1$s on the lower triangular parts and $1$s elsewhere, i.e., $\bm{C}_{k,l}=-1,k>l$ and $\bm{C}_{k,l}=1,k\le l$. Then, the optimal vector $\bm{v}^{t+1}\in\mathbb{C}^N$ in the $(t+1)$-th iteration of ADMM that minimizes $\mathcal{L}_{\rho}$ is given by
	\begin{align}\label{update_v}
		&{\bm{v}}^{t+1}=\bm{g}\odot \Bigg(-\frac{\bm{e}_1}{\rho}+\mathcal{T}^{{\rm Adj}}(\bm{\Psi}^{R}_0)+\mathcal{T}_1^{{\rm Adj}}(\bm{\Psi}^{I}_0)+\nonumber\\
		&\frac{\mathcal{T}^{{\rm Adj}}(\bm{\Lambda}_0^R+\bm{C}\odot \bm{\Lambda}_0^I)}{\rho}\Bigg),
	\end{align}
	where $\mathcal{T}_1(\bm{z}):=\bm{C}\circ \mathcal{T}(\bm{z})$ for an arbitrary vector $\bm{z}\in\mathbb{C}^{N}$ and $\mathcal{T}_1^{{\rm Adj}}$ is the adjoint operator of $\mathcal{T}_1$.
\end{lem}
\begin{proof}
See Appendix \ref{proof.lem.v_cal}.
\end{proof}

The update of $\bm{\Psi}$ requires to compute a projection onto the cone of positive semidefinite Hermitian matrices which can be regarded as the most costly part and is provided below:
\begin{eqnarray}\label{eq.Psi_update}
	\bm{\Psi}^{t+1}&=&\mathop{\arg\min}_{\bm{\Psi}\succeq \bm{0}} \Bigg\|\begin{bmatrix}
		\bm{\Psi}_0&\bm{\Psi}_1\\
		\bm{\Psi}_1^H&\overline{\bm{\Psi}}
	\end{bmatrix}-\begin{bmatrix}
		\mathcal{T}(\bm{v})&\bm{Z}\\
		\bm{Z}^H&\bm{W}
	\end{bmatrix}\nonumber\\
	&+&\frac{1}{\rho}\begin{bmatrix}
		\bm{\Lambda}_0&\bm{\Lambda}_1\\
		\bm{\Lambda}_1 & \overline{\bm{\Lambda}}
	\end{bmatrix}\Bigg\|_F^2. 
\end{eqnarray} 
The relation \eqref{eq.Psi_update} can be solved by computing the eigenvalue decomposition of $\begin{bmatrix}
	\mathcal{T}(\bm{v})&\bm{Z}\\
	\bm{Z}^H&\bm{W}
\end{bmatrix}-\frac{\bm{\Lambda}}{\rho}$ and retaining only the directions corresponding to positive eigenvalues. By performing the steps \eqref{update_Z&W}, \eqref{update_v}, \eqref{eq.Psi_update} and \eqref{Lambda_update} iteratively, the dual matrix multiplier $\bm{\Lambda}_1\in\mathbb{C}^{N\times T}$ is found. However, according to Theorem \ref{thm.main}, to achieve the AoAs corresponding to the active users, we need an estimate for the dual matrix $\bm{V}\in\mathbb{C}^{M\times T}$. In the following lemma whose proof is provided in Appendix \ref{proof.lem.dual_rel}, we find a closed-form relation that clearly specifies the connection between these variables. 
\begin{lem}\label{lem.dual_var+dual_multiplier}
	The dual solution corresponding to the primal variable $\bm{Y}^{\star}$ denoted by $\bm{V}\in\mathbb{C}^{N\times T}$ is linked with the dual matrix multiplier $\bm{\Lambda}_1\in\mathbb{C}^{N\times T}$ obtained in the proposed ADMM algorithm in Subsection \ref{sec.admm} via the following closed-form relation
	\begin{align}
		\mathcal{P}^{{\rm Adj}}(\bm{V})=\frac{\bm{\Lambda}_1}{c_1}.
	\end{align} 
\end{lem}

%    First, we obtain closed form solutions for $\bm{Z}^{t+1}$ and $\bm{W}^{t+1}$ which are obtained by

\subsection{Identification of Active Users}\label{sec.active_user}
Once $\mathcal{P}^{{\rm Adj}}(\bm{V})$ is found using Lemma \ref{lem.dual_var+dual_multiplier}, we can use the relation \eqref{eq.angle_find} in Theorem \ref{thm.main} to estimate the active angles corresponding to the active devices. Although we have now full access to the angles, there is still ambiguity in which angles corresponds to which active user. Applying a clustering algorithm, we can place angles into several clusters, each of which represents the angles of an active device. If the number of active devices $K_a$ is known beforehand (an assumption widely used in the literature, e.g., \cite{ke2020compressive}), the $k$-means clustering algorithm can be employed with known number of clusters. When $K_a$ is not known, which seems to be more reasonable in practice, we can use \textit{elbow} clustering method \cite{liu2020determine}, which chooses the optimal number of clusters by identifying a sharp elbow (knee point) in the graph of explained variations versus clusters. The key idea is that adding another cluster does not provide much better modeling of the angles. The elbow clustering algorithm provides the angles estimate
$\{\widehat{\bm{\theta}}^k\}_{k=1}^{\widehat{K}_a}$ in $\widehat{K}_a$ clusters as follows
\begin{align}
	[\{\widehat{\bm{\theta}}^k\}_{k=1}^{\widehat{K}_a},\widehat{K}_a,\{\widehat{L}_k\}_{k=1}^{\widehat{K}_a}]={\rm elbow}(\bm{\widehat{\theta}}).    
\end{align}
The number of clusters provides an estimate of $K_a$ and the number of angles within the $k$-th cluster gives an estimate of $\widehat{L}_k$. 

If the devices are stationary and the BS have full information on the angles (for instance only the line of sight (LoS) angle) of each device before random access, then the BS can exactly understand at this stage which device is active and may provide an estimate for $\mathcal{S}_{\rm AU}$, denoted by $\widehat{\mathcal{S}}_{\rm AU}$, and the RA task is finished. If angle information of the devices is not provided or devices do not remain stationary in a known place, the BS can still identify the active users by recovering their unique preambles; this is explained in the next subsection.

\subsection{Joint Data Recovery and Channel Estimation}\label{sec.data_rec}
To recover the preamble and data of each device, we employ an alternating minimization scheme, in which an initial estimate of $\bm{\phi}_k$ is first generated and the following optimization is then solved as a means to find the complex channel amplitudes:
\begin{align}\label{eq.c_estimate}
	[\widehat{\bm{\alpha}}^1,..., \widehat{\bm{\alpha}}^{K_a}]=\mathop{\arg\min}_{\substack{{\alpha}_1^k,..., \alpha^k_{\widehat{L}_k},\\k=1,..., \widehat{K}_a}}\|\bm{Y}-\sum_{k\in\widehat{\mathcal{S}}_{\rm AU}}\sum_{l=1}^{\widehat{L}_k}\alpha_l^k\mathcal{P}_{\Omega}(\bm{a}(\widehat{\theta}_l^k)\widehat{\bm{\phi}}_k^H)\|_F
\end{align}
which can be obtained by 
\begin{align}\label{alpha_hat}
	\widehat{\bm{\alpha}}^k=(\mathcal{P}_{\Omega}\bm{A}_k)^{\dagger} \bm{Y} \bm{\phi}^{\dagger}_k
\end{align}
where $\bm{A}_k:=[\bm{a}(\theta^k_1),..., \bm{a}(\theta^k_{\widehat{L}_k})]\in\mathbb{C}^{N\times \widehat{{L}}_k}$ and $\bm{\phi}^{\dagger}_k$ is the $k$-th column of the pseudo inverse of $\bm{\Phi}:=[\bm{\phi}_1^T,...,\bm{\phi}_{\widehat{K}_a}^T]^T\in\mathbb{C}^{\widehat{K}_a\times T}$ denoted by ${\bm{\Phi}}^{\dagger}$.

Having now an estimate of $\widehat{\bm{\alpha}}^k$, we can update the estimate $\widehat{\bm{\phi}}_k$ by solving 
\begin{align}
	[\widehat{\bm{\phi}}_1,...,\widehat{\bm{\phi}}_{\widehat{K}_a}]=\mathop{\arg\min}_{\substack{\phi_1,..., \phi_{\widehat{K}_a}\\\bm{\phi}_k\ge \bm{0}}}\|\bm{Y}-\sum_{k\in\widehat{\mathcal{S}}_{\rm AU}}\sum_{l=1}^{\widehat{L}_k}\widehat{\alpha}_l^k\mathcal{P}_{\Omega}(\bm{a}(\widehat{\theta}_l^k){\bm{\phi}}_k^H)\|_F
\end{align}
where the non-negative constraint $\bm{\phi}_k\ge \bm{0}$ is included to remove any ambiguities caused by multiplications of two elements, as well as to uniquely identify the preamble and data of UEs. The latter optimization problem leads to the following closed-form relation
\begin{align}\label{phi_hat}
	\widehat{\bm{\Phi}}=\Big(\bm{B}^{\dagger} \bm{Y}\Big)_{+}\in\mathbb{C}^{\widehat{K}_a\times T}
	\end{align}
where $\bm{B} \coloneqq [\mathcal{P}_{\Omega}\bm{A}_1 \widehat{\bm{\alpha}}^1,...,\mathcal{P}_{\Omega}\bm{A}_{\widehat{K}_a} \widehat{\bm{\alpha}}^{\widehat{K}_a}]\in\mathbb{C}^{M\times \widehat{K}_a}$ and $(\cdot)_+$ keeps the non-negative parts of a matrix as they are and makes  zero elsewhere. By performing steps \eqref{alpha_hat} and \eqref{phi_hat} alternatively, the data and complex channel coefficients are uniquely found within few iterations. Since the transmitted data contains the unique user
ID, the active user indices are obtained after recovering the users' data $\bm{\phi}_k$s.

The pseudocode of the proposed method, which summarizes the aforementioned steps, is provided in Algorithm \ref{algorithm.admm}.

%\begin{figure}[htb]
{\centering
	%	\resizebox{.5\textwidth}{!}{
		\begin{minipage}{.48\textwidth}
			\begin{algorithm}[H]
				\caption{Blind Goal-Oriented Detection (BGOD)}
				%	$ RMSPI and  GRMSPI $
				\begin{algorithmic}[1]\label{algorithm.admm}
					\REQUIRE $\bm{Y}\in\mathbb{C}^{M\times T}$, ${\rm iter}^{ADMM}_{\max}$, ${\rm iter}^{AM}_{\max}$
					\STATE Initialize the iteration index as $t=0$.
					\STATE Initialize the data vectors as $\bm{\phi}_k \sim \mathcal{CN}(\bm{0},\bm{I}_T), k=1,..., K_a$. 
					%		\STATE $y \leftarrow 1$
					\STATE Set $\bm{\Psi}^0=\bm{\Lambda}^0=\bm{0}$.
					\WHILE {$t\le {\rm iter}^{ADMM}_{\max}$}
					\STATE Obtain $\bm{Z}^{t+1}$, $\bm{W}^{t+1}$ and $\bm{v}^{t+1}$ according to \eqref{update_Z&W} and \eqref{update_v},
					
					%					 \eqref{eq.Psi_update} and \eqref{Lambda_update}, respectively.
					\STATE Compute the eigenvalue decomposition of the matrix $\begin{bmatrix}
						\mathcal{T}(\bm{v})&\bm{Z}\\
						\bm{Z}^H&\bm{W}
					\end{bmatrix}-\frac{\bm{\Lambda}}{\rho}$ in \eqref{eq.Psi_update} and set all non-negative eigenvalues to zero in order to obtain $\bm{\Psi}^{t+1}$ in \eqref{eq.Psi_update}.
					\STATE Update the dual Lagrange matrix multiplier $\bm{\Lambda}^{t+1}$ according to \eqref{Lambda_update}.
					\STATE $t\rightarrow t+1$
					\ENDWHILE
					\STATE Obtain $\mathcal{P}^{{\rm Adj}}(\widehat{\bm{V}})$ using Lemma \ref{lem.dual_var+dual_multiplier}.
					\STATE Localize the angles by discretizing $\theta$ on a fine grid up to a desired accuracy and identify where the $\ell_2$ norm of the polynomial $\bm{q}_{G}(\theta)$ achieves to its maximum as in Theorem \ref{thm.main}.
					\STATE \text{Perform the elbow method:} $[\{\widehat{\bm{\theta}}^k\}_{k\in \widehat{\mathcal{S}}_{\rm AU}},\widehat{K}_a,\{\widehat{L}_k\}_{k\in \widehat{\mathcal{S}}_{\rm AU}}]={\rm elbow}(\bm{\widehat{\theta}})$
					. 
					\FOR {it=1 to ${\rm iter}^{AM}_{max}$}
					\STATE Compute the complex channel coefficients $\widehat{\bm{\alpha}}^k$ using \eqref{alpha_hat}.   
					\STATE Compute the data estimate via \eqref{phi_hat}.
					\ENDFOR				
					
				\end{algorithmic} 
				Return: $\widehat{\mathcal{S}}_{\rm AU}, \{\widehat{\bm{\theta}^k}\}_{ k\in\widehat{\mathcal{S}}_{\rm AU}},\{\widehat{\bm{\phi}_k}\}_{ k\in\widehat{\mathcal{S}}_{\rm AU}}, \{\widehat{\bm{\alpha}^k}\}_{ k\in\widehat{\mathcal{S}}_{\rm AU}}, \widehat{K}_a= |\widehat{\mathcal{S}}_{\rm AU}|, \{\widehat{L}_k\}_{ k\in\widehat{\mathcal{S}}_{\rm AU}}, \widehat{\bm{h}}_k=\sum_{l=1}^{\widehat{L}_k}\widehat{\alpha}_l^k\bm{a}(\widehat{\theta}_l^k), k\in\widehat{\mathcal{S}}_{\rm AU}$. 
			\end{algorithm} 
		\end{minipage}
	}
	%}
%\end{figure}
%\begin{figure}[t]
%	\hspace{-.4cm}
%	\includegraphics[scale=0.37]{rapblock.pdf}
%	\caption{{\color{\change} Data transmission procedure of cMBB. Step A is dedicated to RAP for AUE and CE and step B is dedicated to coherent data transmission.}}\label{fig.rap}
%\end{figure}

%Our algorithm works for crowded scenarios. Therefore, our scheme for RA, is that the active users send arbitrary pilots in $T$ time slots. The receiver estimates the users' channels and users' identity exactly and without any ambiguity with the lowest overhead. Then, the users send their data through the channel and the BS estimates the data of users by \eqref{eq.phi_estimate}. Unlike previous works in RA, there are no limitations in the number of active users and consequently the number of pilots. 

\section{Simulation Results}\label{sec.simulations}
In this section, we provide numerical results to assess the performance of BGOD and compare it with state-of-the-art RA schemes, namely covariance-based \cite{haghighatshoar2018improved,haghighatshoar2018new,fengler2019massive} and AMP approaches \cite{ke2020compressive}. For that, we designed optimal pilot settings and distributions for these works, as exactly mentioned in their references. This is a challenging (or unfair) setting for our proposed scheme, since $\bm{\phi}_k$s serve as pilots for the aforementioned methods considered perfectly known at the BS, while in the proposed BGOD, $\phi_k$s are users' data and are fully unknown to the BS. The RA performance is quantified in terms of detection and false alarm probabilities, defined by $\mathds{P}_d=\frac{\mathds{E}|\mathcal{S}_{\rm AU}\bigcap \widehat{\mathcal{S}}_{\rm AU}|}{K_a}$ and $\mathds{P}_{fa}=\frac{\mathds{E}|\widehat{\mathcal{S}}_{\rm AU} \setminus \mathcal{S}_{\rm AU}|}{K-K_a}$, respectively, where $|\widehat{\mathcal{S}}_{\rm AU} \setminus \mathcal{S}_{\rm AU}|$ counts the number of elements in $\widehat{\mathcal{S}}_{\rm AU}$ that are not present in $\mathcal{S}_{\rm AU}$. Channel estimation and data recovery performance is evaluated by the normalized mean square error (NMSE) provided respectively by ${\rm NMSE}_{\rm CE}=\mathds{E}\frac{\|\widehat{\bm{H}}_E-\bm{H}_E\|_F}{\|\bm{H}_E\|_F}$ and ${\rm NMSE}_{\rm DR}=\mathds{E}\frac{\|\widehat{\bm{\Phi}}_E-\bm{\Phi}_E\|_F}{\|\bm{\Phi}_E\|_F}$. Here, $\bm{H}_E\in\mathbb{C}^{N\times K}$ and $\widehat{\bm{H}}_E\in\mathbb{C}^{N\times K}$ are extended matrices with columns $\bm{h}_k, k\in \mathcal{S}_{\rm AU}$ and $\widehat{\bm{h}}_k, k\in \widehat{\mathcal{S}}_{\rm AU}$, respectively, and are zero elsewhere. Similarly, $\bm{\Phi}_E\in\mathbb{C}^{K\times T}$ and $\widehat{\bm{\Phi}}_E\in\mathbb{C}^{K\times T}$ are extended matrices whose rows in the indices $\mathcal{S}_{\rm AU}$ and $\widehat{\mathcal{S}}_{\rm AU}$ are given by $\bm{\Phi}\in\mathbb{C}^{K_a\times T}$ and $\widehat{\bm{\Phi}}\in\mathbb{C}^{\widehat{K}_a\times T}$, respectively, and are zero elsewhere. To calculate $\mathds{P}_{d}$, $\mathds{P}_{fa}$ and ${\rm NMSE}$, we perform $50$ Monte Carlo iterations in our experiments. The computational complexity is also averaged over $50$ realizations in all experiments. The regularization parameter $\gamma$ is set to $\frac{1}{\sigma^2}$, where $\sigma^2=\mathds{E}|E_{i,l}|^2, \forall i, l$ is the variance of each noise element. We define the signal to noise ratio (in ${\rm dB}$) by ${\rm SNR}=10\log_{10}(\frac{\|\mathcal{P}_{\Omega}(\bm{X})\|_F^2}{MT\sigma^2})$. To fully exploit the BS antenna capabilities,  which is also the case in AMP and covariance-based methods, we assume first that the BS observes the output of all of its antennas, i.e., $\Omega=\{1,...N\}$ with $M=N$. 

Figure \ref{fig.detection_versus_N} shows the effect of the number $N$ of BS antennas on the performance, with parameters $K=100$, $K_a=3$, $L_{\max}=3$, and ${\rm SNR}=30 {\rm dB}$, when the least number of time resources are employed at the devices (i.e., $T=1$). We observe that the detection probability of BGOD tends to one while the false alarm probability gets zero for increasing $N$. 
%Since when BGOD starts to detect perfectly, the computational complexity increases with $N$. 
In contrast, AMP \cite{ke2020compressive} performs poorly in that case while the covariance-based method \cite{haghighatshoar2018improved,fengler2019massive} does not provide acceptable detection performance. 
Figure \ref{fig.detection_versus_Ka} depicts the performance in a scenario where the number $K_a$ of active devices increases while everything else is kept fixed at $N=64$, $K=100$, $L_{\max}=3$, $T=1$, and ${\rm SNR}=30{\rm dB}$. BGOD shows superior performance in terms of detection and false alarm probability compared to AMP and covariance-based methods. However, the computational complexity of BGOD is slightly higher than existing RA approaches. In Figure \ref{fig.detection_versus_T} we study the effect of increasing the time resources $T$, while keeping the other parameters fixed, as $N=64$, $K=100$, $K_a=3$, $L_{\max}=3$, and ${\rm SNR}=50{\rm dB}$. Although the detection performance of AMP and covariance-based methods improves when the number of known pilots increases, it remains inferior compared to BGOD, whose computational complexity though increases with $T$. It should be mentioned that the additional computational complexity comes from the fact that BGOD simultaneously performs active device detection, channel estimation and data recovery, while AMP does only active detection and channel estimation, and covariance-based carries out only active detection. 
Figure \ref{fig.M_change} shows the effect of the number $M$ of observed elements at the BS for $N=128$, $T=1$, ${\rm SNR}=30{\rm dB}$, $L_{\max}=3$, $K_a=20$, and $K=500$. We see that observing only few array elements is sufficient to achieve very good detection performance. Moreover, the complexity is not affected by $M$ when BGOD starts to detect active devices perfectly. 

Finally, we evaluate CE and DR performance of BGOD. In Figure \ref{fig.CE_performance}, we observe that BGOD estimates very well the channels corresponding to the active users while AMP method exhibits poor CE performance. Note that this result is obtained for AMP knowing in advance all pilot sequences $\bm{\phi}_k$s with optimal distribution settings as in \cite{ke2020compressive}, whereas BGOD works in a blind way and not only estimates channels but also recovers the data of active users simultaneously. This also justifies the additional time BGOD requires, as seen at the right side in Figure \ref{fig.CE_performance}.
\begin{figure}[t]
	\hspace{-0cm}
	\begin{subfigmatrix}{2}
		\subfigure[]{\includegraphics[scale=1,trim={0cm 0cm 0cm 0cm}]{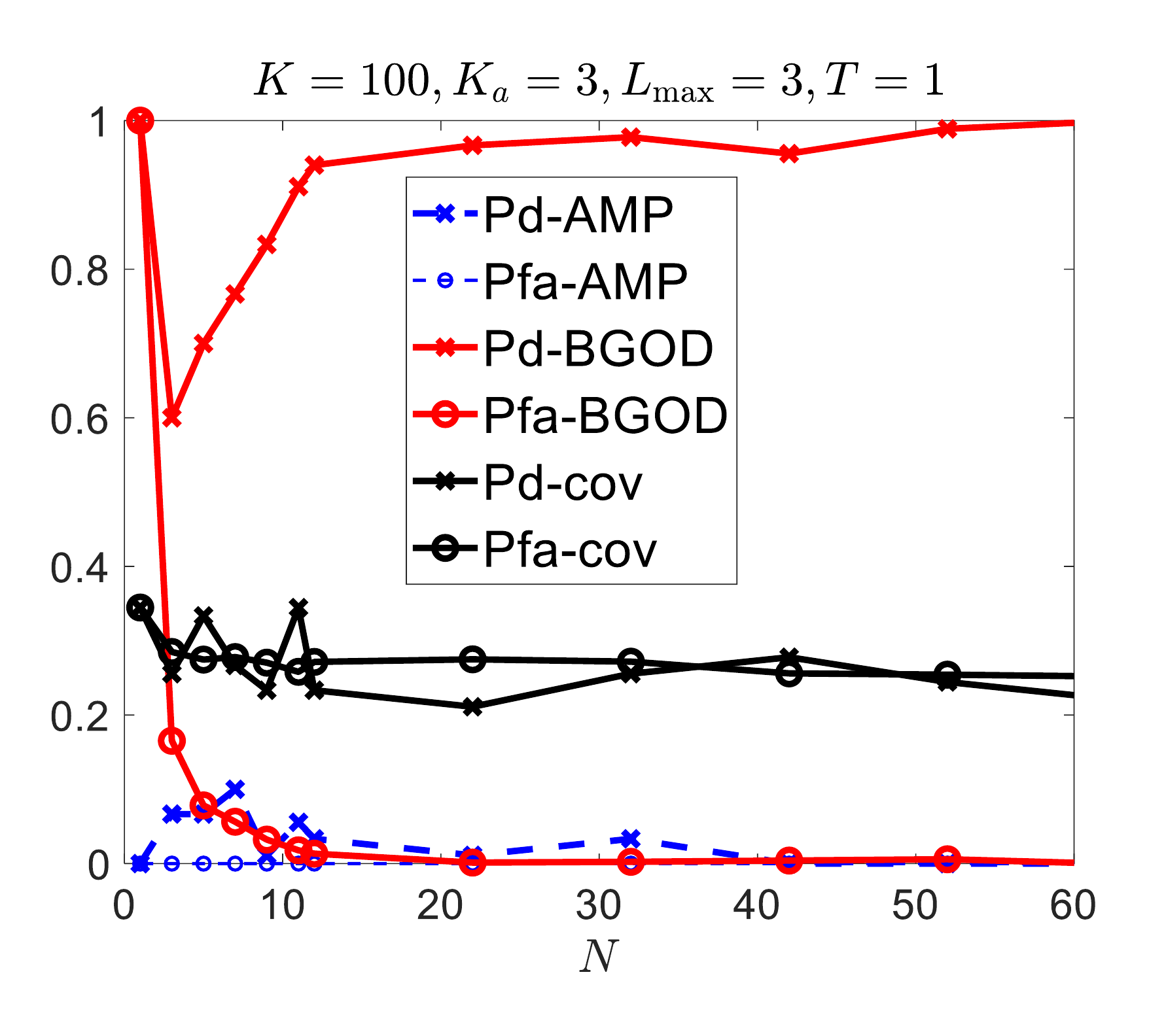}}
		\subfigure[]{\includegraphics[scale=1,trim={0cm 0cm 0cm 0cm}]{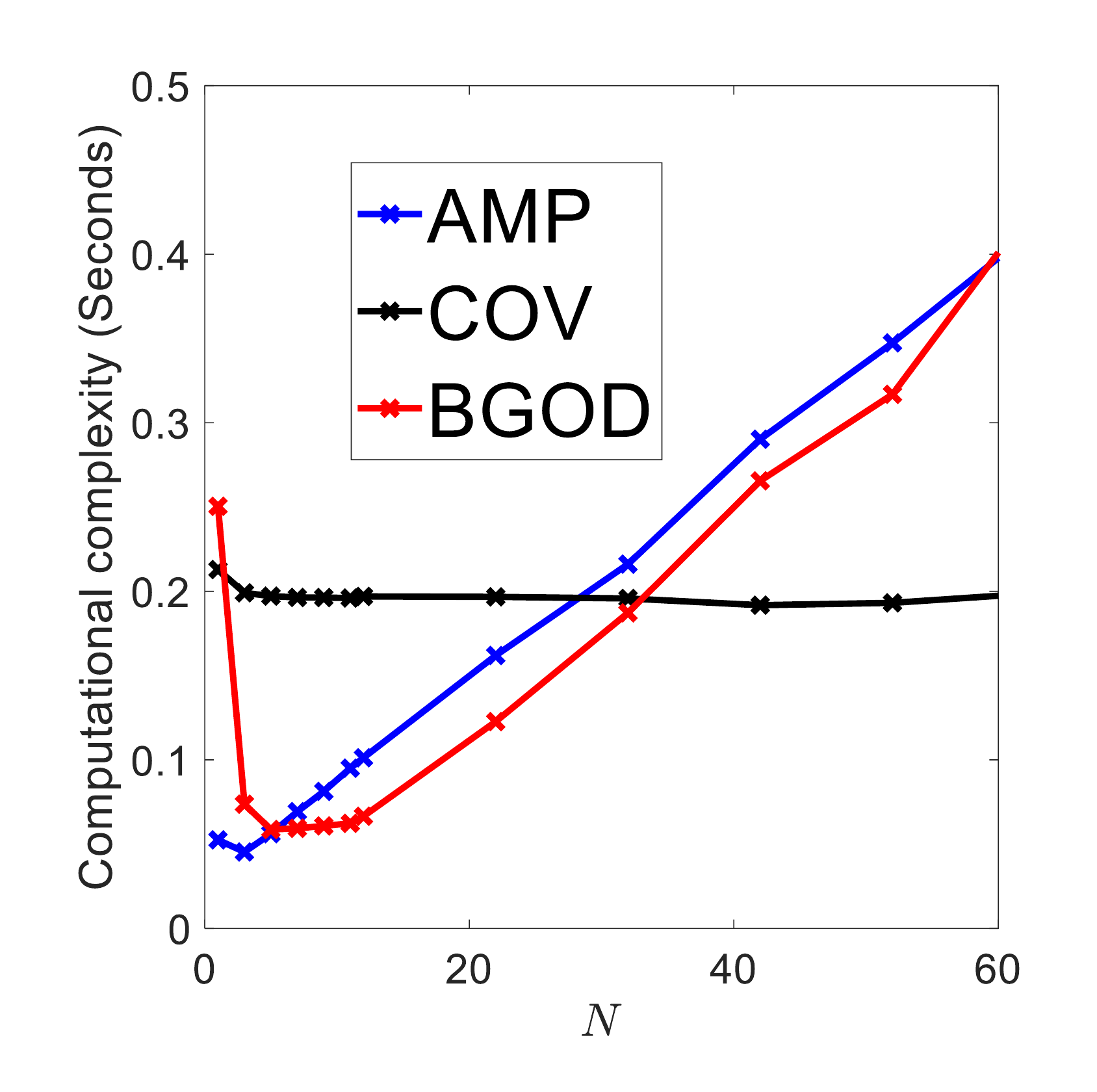}}
	\end{subfigmatrix}
	%	\mbox{\subfigure[]{\includegraphics[scale=.5,trim={0cm 0cm 0cm 0cm}]{main_fig1.tex}}\quad
		%		%		\subfigure[]{\includegraphics[width=2.34in]{paperfigR9.pdf}\label{fig.paperfigR9}}\quad
		%		
		%		\subfigure[]{\includegraphics[scale=.5,trim={0cm 0cm 0cm 0cm}]{main_fig2.tex}}}
	%	\input{main_fig1.tex}
	%	{\includegraphics[scale=.7,trim={0cm 0cm 0cm 0cm}]{main_fig1.tex}}
	%	\input{channel.tex}
	%	\includegraphics[scale=0.48]{channel.png}
	\caption{Detection performance comparison between the proposed BGOD method and AMP and covariance-based schemes versus the values of $N$ for $K=100, K_a=3, L_{\max}=3, T=1$, and ${\rm SNR}=30 {\rm dB}$. (a) Detection accuracy (b) Computational complexity.}\label{fig.detection_versus_N}
\end{figure}
\begin{figure}[t]
	\hspace{-0cm}
	\begin{subfigmatrix}{2}
		\subfigure[]{\includegraphics[]{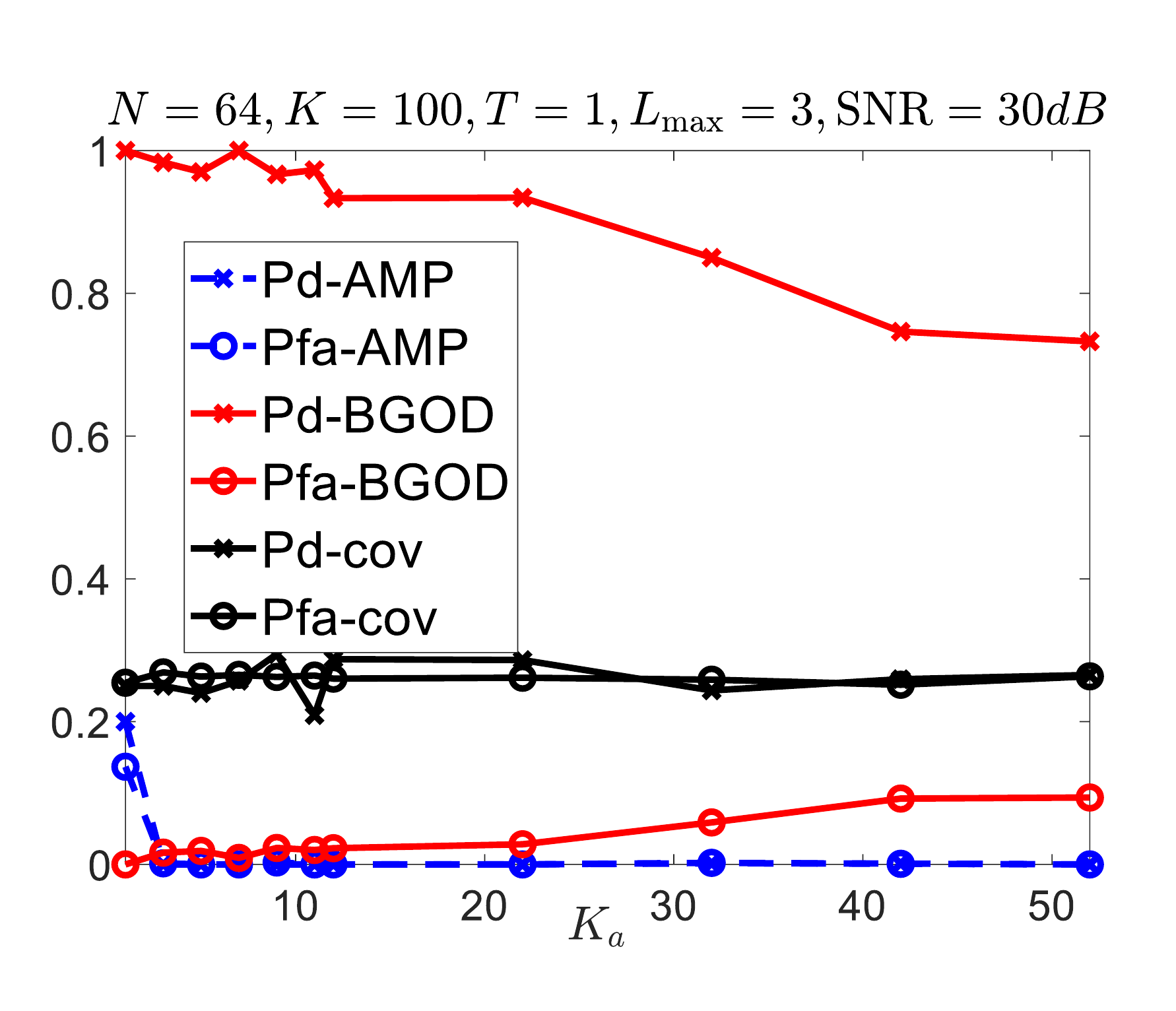}}
		\subfigure[]{\includegraphics[]{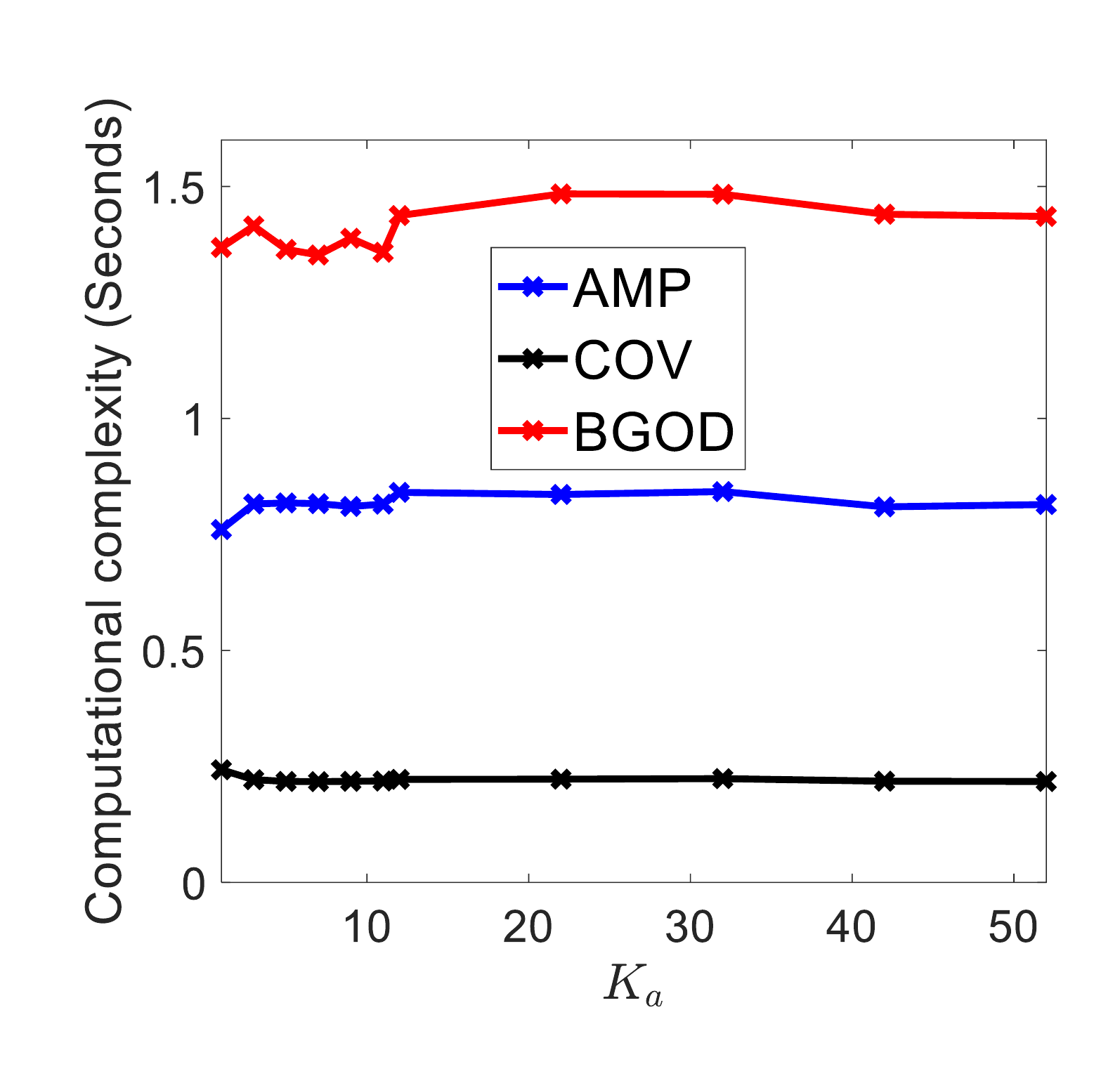}}
	\end{subfigmatrix}
	%	\mbox{\subfigure[]{\includegraphics[scale=.5,trim={0cm 0cm 0cm 0cm}]{main_fig1.tex}}\quad
		%		%		\subfigure[]{\includegraphics[width=2.34in]{paperfigR9.pdf}\label{fig.paperfigR9}}\quad
		%		
		%		\subfigure[]{\includegraphics[scale=.5,trim={0cm 0cm 0cm 0cm}]{main_fig2.tex}}}
	%	\input{main_fig1.tex}
	%	{\includegraphics[scale=.7,trim={0cm 0cm 0cm 0cm}]{main_fig1.tex}}
	%	\input{channel.tex}
	%	\includegraphics[scale=0.48]{channel.png}
	\caption{Detection performance comparison between the proposed BGOD method and AMP and covariance-based methods versus the number of active devices $K_a$ for $N=64, K=100, L_{\max}=3, T=1$, and ${\rm SNR}=30 {\rm dB}$. (a) Detection accuracy (b) Computational complexity.}\label{fig.detection_versus_Ka}
\end{figure}
\begin{figure}[t]
	\hspace{-0cm}
	\begin{subfigmatrix}{2}
		\subfigure[]{\includegraphics[]{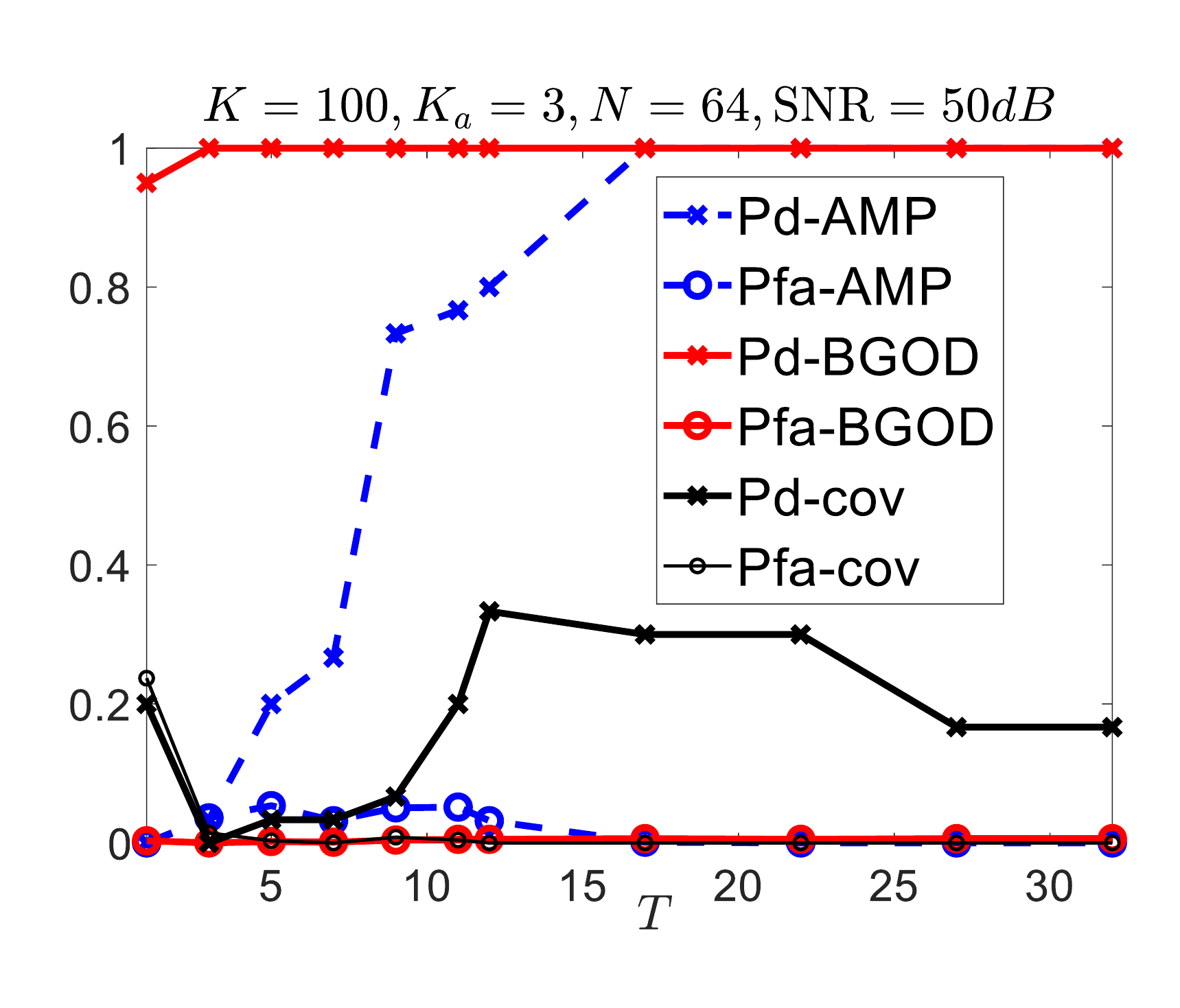}}
		\subfigure[]{\includegraphics[]{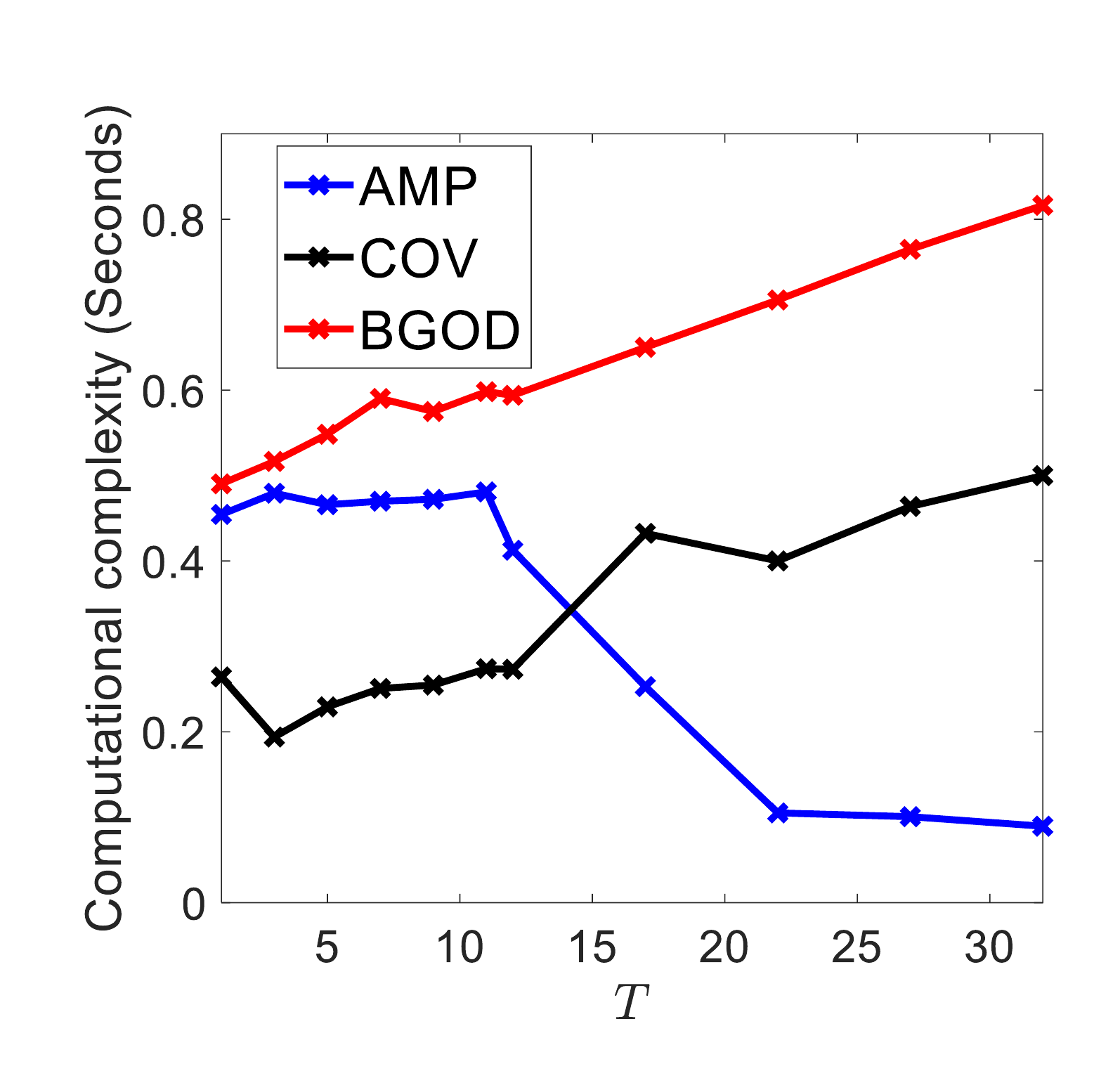}}
	\end{subfigmatrix}
	%	\mbox{\subfigure[]{\includegraphics[scale=.5,trim={0cm 0cm 0cm 0cm}]{main_fig1.tex}}\quad
		%		%		\subfigure[]{\includegraphics[width=2.34in]{paperfigR9.pdf}\label{fig.paperfigR9}}\quad
		%		
		%		\subfigure[]{\includegraphics[scale=.5,trim={0cm 0cm 0cm 0cm}]{main_fig2.tex}}}
	%	\input{main_fig1.tex}
	%	{\includegraphics[scale=.7,trim={0cm 0cm 0cm 0cm}]{main_fig1.tex}}
	%	\input{channel.tex}
	%	\includegraphics[scale=0.48]{channel.png}
	\caption{Detection performance comparison between the proposed BGOD method and AMP and covariance-based schemes versus the number of preambles $T$ for $N=64, K=100, K_a=3, L_{\max}=3$, and ${\rm SNR}=50 {\rm dB}$. (a) Detection accuracy (b) Computational complexity.}\label{fig.detection_versus_T}
\end{figure}
\begin{figure}[t]
	\hspace{-0cm}
	\begin{subfigmatrix}{2}
		\subfigure[]{\includegraphics[]{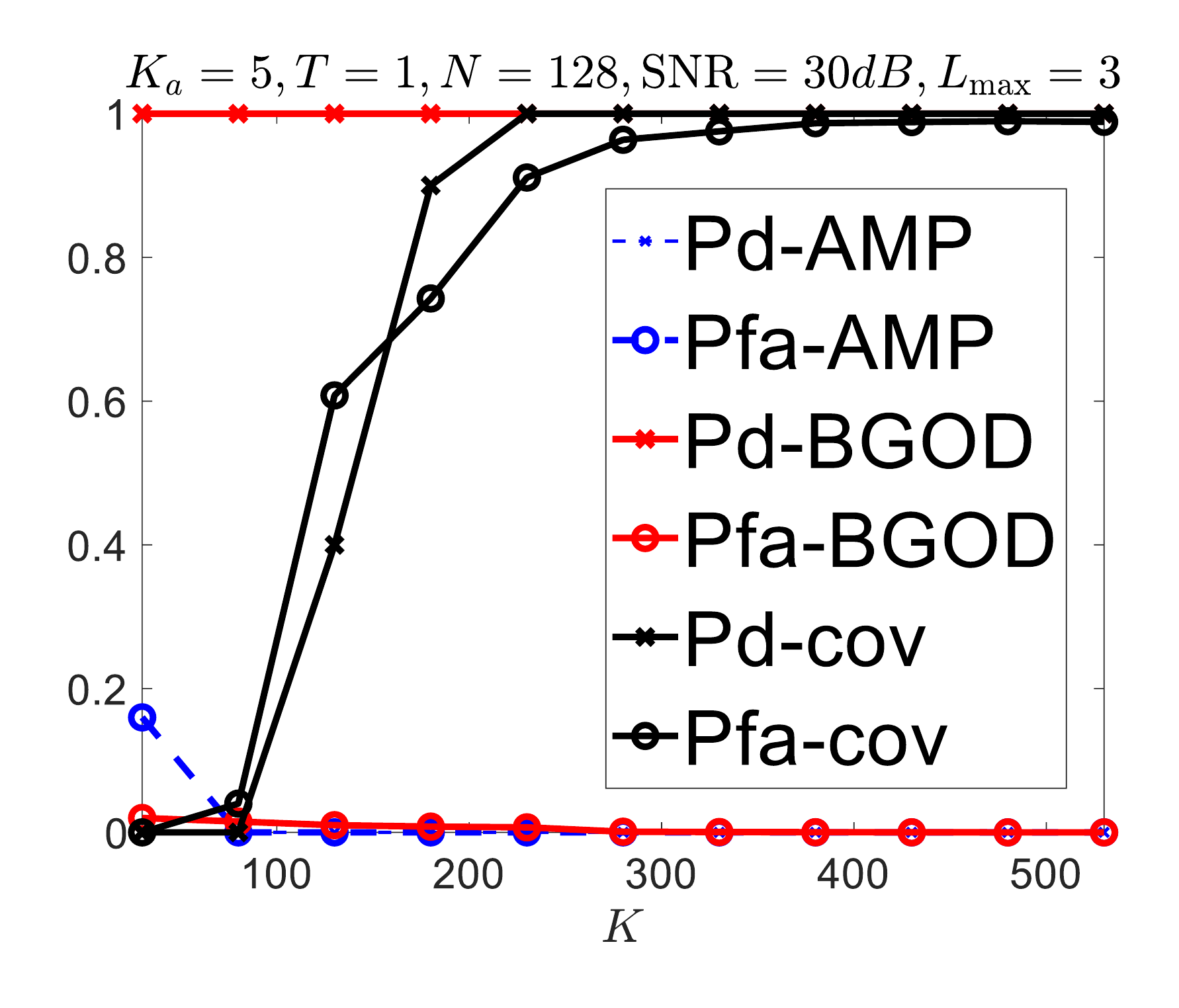}}
		\subfigure[]{\includegraphics[]{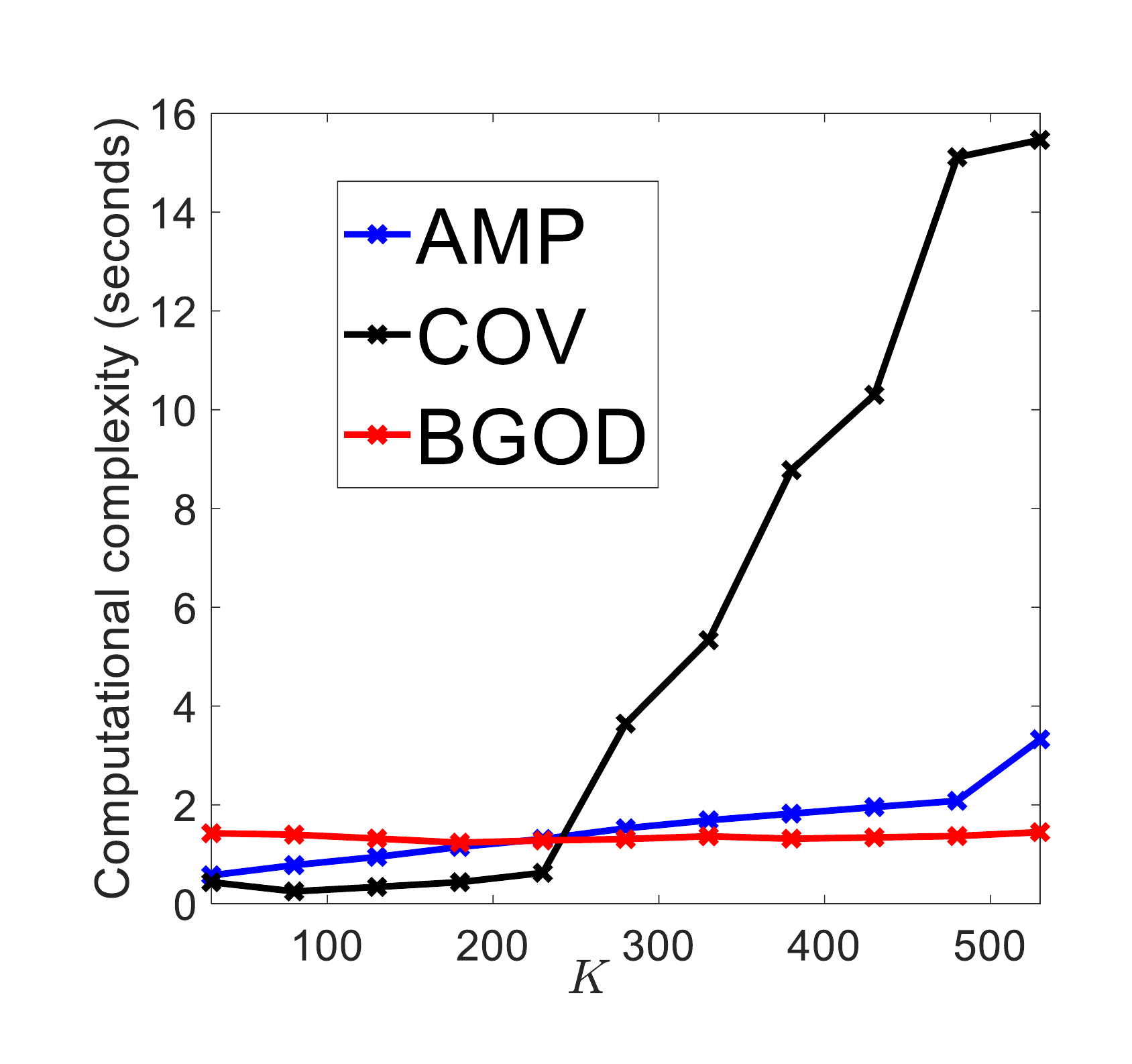}}
	\end{subfigmatrix}
	%	\mbox{\subfigure[]{\includegraphics[scale=.5,trim={0cm 0cm 0cm 0cm}]{main_fig1.tex}}\quad
		%		%		\subfigure[]{\includegraphics[width=2.34in]{paperfigR9.pdf}\label{fig.paperfigR9}}\quad
		%		
		%		\subfigure[]{\includegraphics[scale=.5,trim={0cm 0cm 0cm 0cm}]{main_fig2.tex}}}
	%	\input{main_fig1.tex}
	%	{\includegraphics[scale=.7,trim={0cm 0cm 0cm 0cm}]{main_fig1.tex}}
	%	\input{channel.tex}
	%	\includegraphics[scale=0.48]{channel.png}
	\caption{Detection performance comparison between the proposed BGOD method and AMP and covariance-based schemes versus the number of devices $K$ for $N=128, K_a=5, T=1, L_{\max}=3$, and ${\rm SNR}=30 {\rm dB}$. (a) Detection accuracy (b) Computational complexity.}\label{fig.detection_versus_K}
\end{figure}
\begin{figure}[t]
	\hspace{-0cm}
	\begin{subfigmatrix}{2}
		\subfigure[]{\includegraphics[]{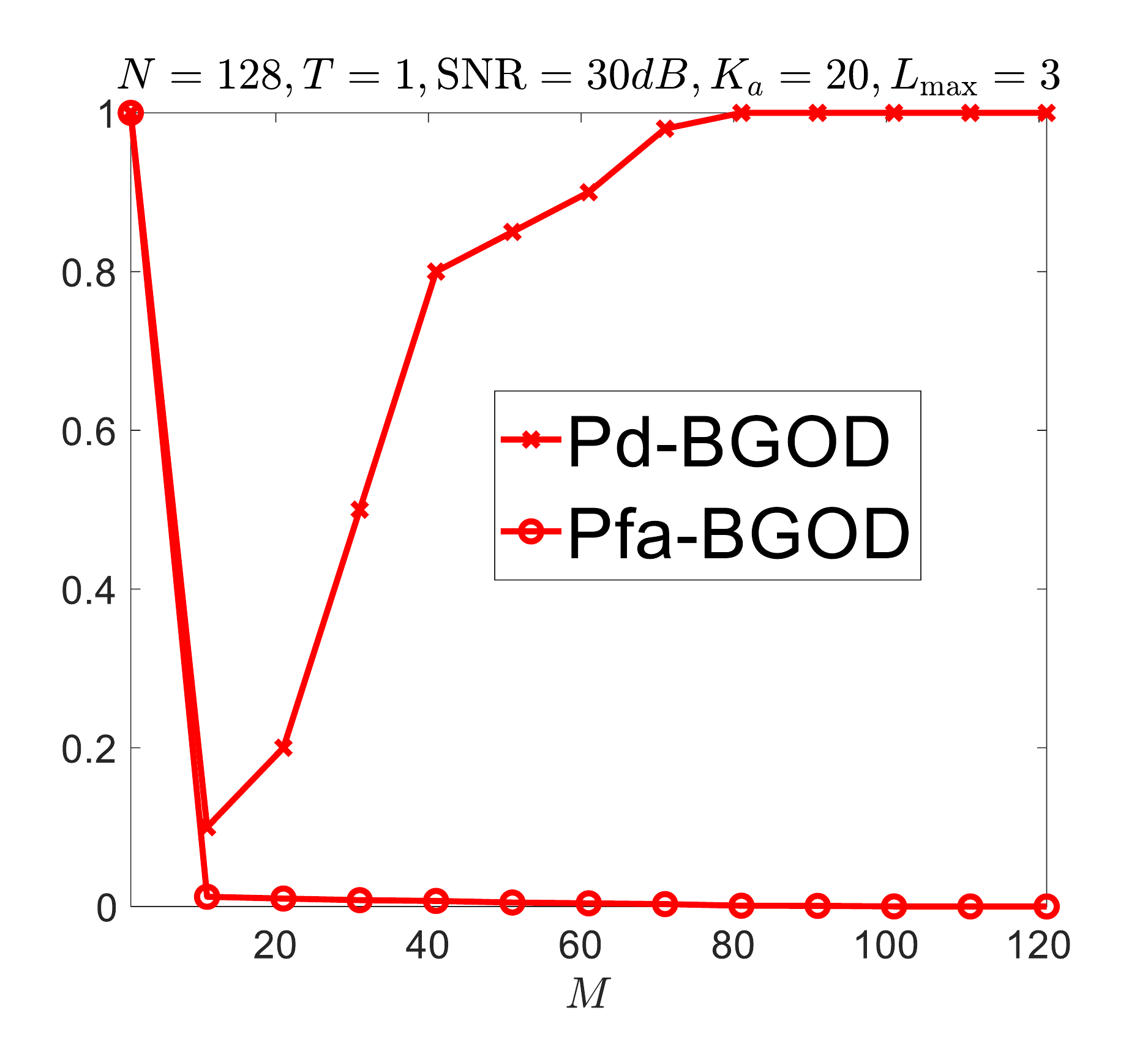}}
		\subfigure[]{\includegraphics[]{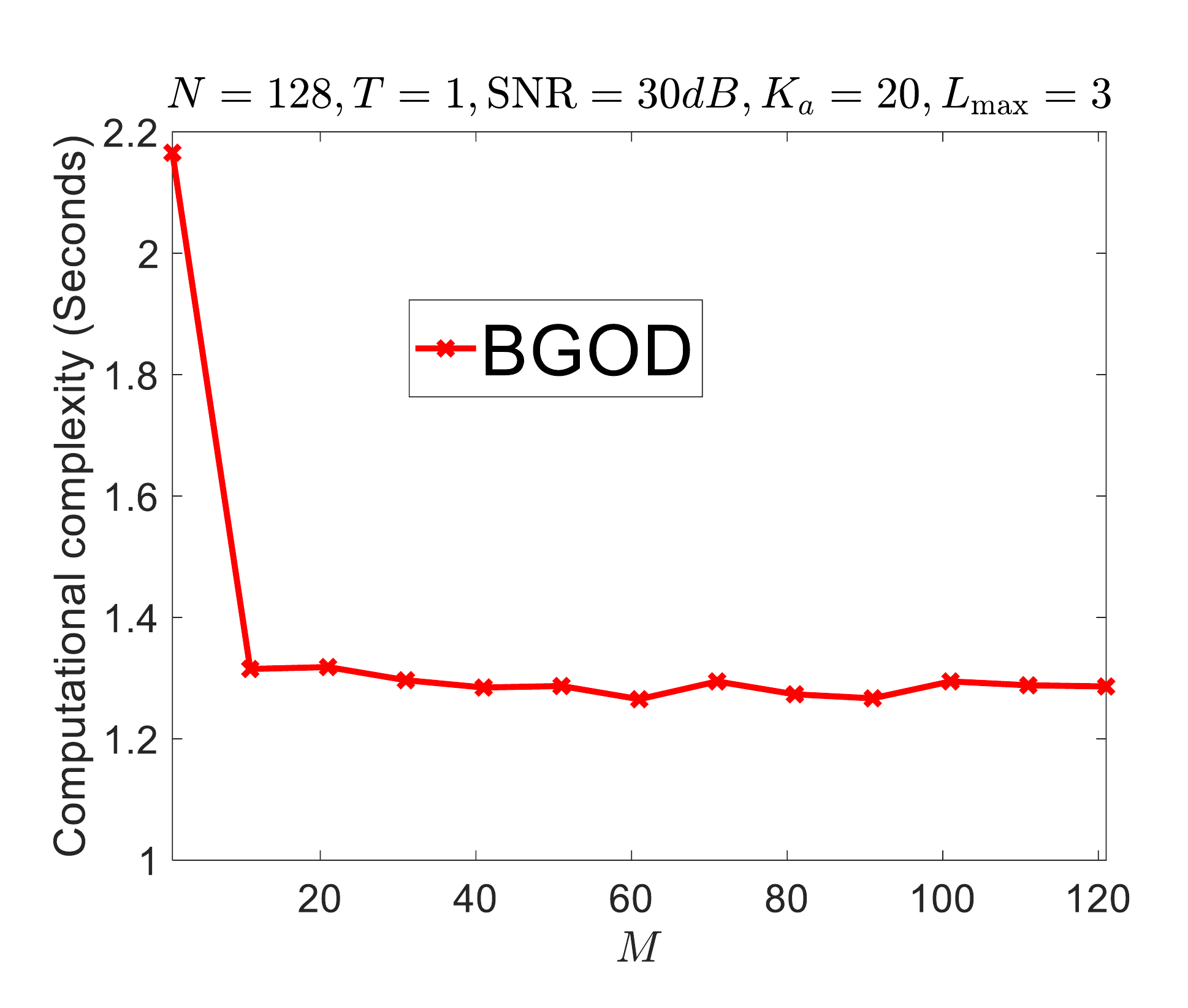}}
	\end{subfigmatrix}
	%	\mbox{\subfigure[]{\includegraphics[scale=.5,trim={0cm 0cm 0cm 0cm}]{main_fig1.tex}}\quad
		%		%		\subfigure[]{\includegraphics[width=2.34in]{paperfigR9.pdf}\label{fig.paperfigR9}}\quad
		%		
		%		\subfigure[]{\includegraphics[scale=.5,trim={0cm 0cm 0cm 0cm}]{main_fig2.tex}}}
	%	\input{main_fig1.tex}
	%	{\includegraphics[scale=.7,trim={0cm 0cm 0cm 0cm}]{main_fig1.tex}}
	%	\input{channel.tex}
	%	\includegraphics[scale=0.48]{channel.png}
	\caption{Random access performance of our proposed method versus the number $M$ of observed arrays at the BS. Figures (a) and (b) show detection accuracy and computational complexity, respectively.}\label{fig.M_change}
\end{figure}
\begin{figure}[t]
	\hspace{-0cm}
	\begin{subfigmatrix}{2}
		\subfigure[]{\includegraphics[]{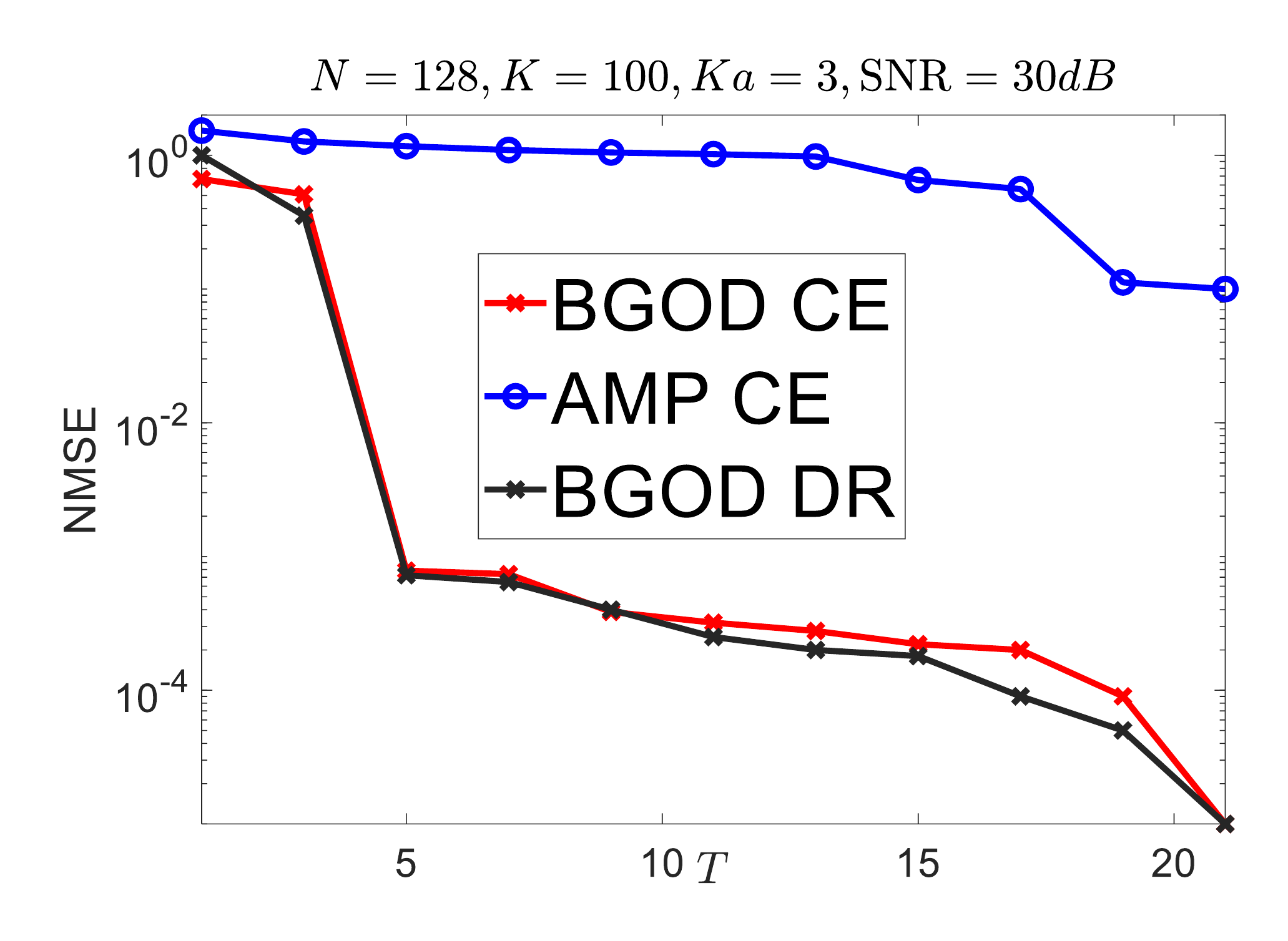}}
		\subfigure[]{\includegraphics[]{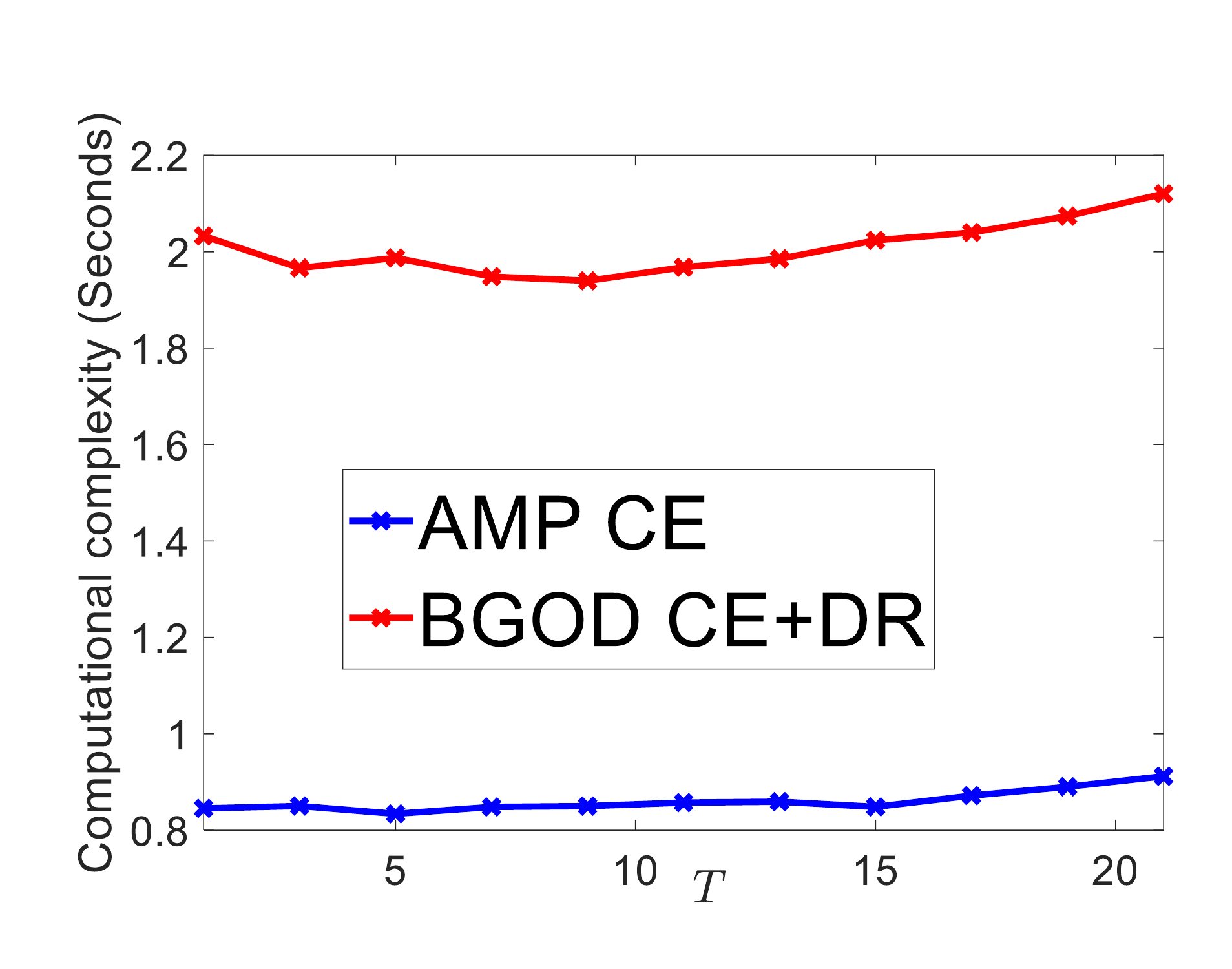}}
	\end{subfigmatrix}
	%	\mbox{\subfigure[]{\includegraphics[scale=.5,trim={0cm 0cm 0cm 0cm}]{main_fig1.tex}}\quad
		%		%		\subfigure[]{\includegraphics[width=2.34in]{paperfigR9.pdf}\label{fig.paperfigR9}}\quad
		%		
		%		\subfigure[]{\includegraphics[scale=.5,trim={0cm 0cm 0cm 0cm}]{main_fig2.tex}}}
	%	\input{main_fig1.tex}
	%	{\includegraphics[scale=.7,trim={0cm 0cm 0cm 0cm}]{main_fig1.tex}}
	%	\input{channel.tex}
	%	\includegraphics[scale=0.48]{channel.png}
	\caption{CE and DR performance of our proposed method versus the time resources $T$ for $N=128, K=100, K_a=5, L_{\max}=3$, and ${\rm SNR}=30 {\rm dB}$. CE performance is compared with AMP and the DR performance of BGOD is shown with cross-dash line. Figures (a) and (b) show accuracy and computational complexity, respectively. The computational complexity of our algorithm relates to both CE and DR simultaneously, whereas only CE is performed in AMP.}\label{fig.CE_performance}
\end{figure}

\section{Conclusion}\label{sec.conclusion}
In this work, we have proposed a novel random access strategy for massive connectivity in future wireless networks. Our scheme is based on a reconstruction-free optimization task, for which we have proposed goal-oriented optimization that helps in finding a relevant continuous function containing sufficient information to obtain the active users identity. In fact, we have proposed a method to achieve the goal of active user detection without reconstructing the corresponding channels and messages, which makes our strategy independent of the number of devices involved. Simulation results have shown the significant performance gains that can be achieved when using our strategy compared to existing state-of-the-art schemes. This makes the proposed blind goal-oriented random access a promising candidate for supporting massive connectivity in future wireless networks. 

%In this work, we proposed a novel strategy for RA in B5G which is based on a reconstruction-free optimization. For this task, we proposed a goal-oriented optimization which helps to find a valuable continuous function. This function has the sufficient information to obtain the ID of active users. In fact, we proposed a way to achieve the goal of active user detection without reconstructing the corresponding channels and messages. This makes our strategy fully independent of the number of devices. As simulation results verify, our method meets all of the URLLC and massive connectivity requirements of B5G and provides superior performance compared to the existing RA methods. 

\appendices
\section{Proof of Theorem \ref{thm.main}}\label{proof.thm_main}
To prove the theorem, we first obtain the dual problem of \eqref{prob.atomic_l1_lasso} in the following: (proved in Appendix \ref{proof.dualprob}.)
\begin{align}\label{prob.dual}
	&\max_{\bm{V}\in\mathbb{C}^{N\times T}}{\rm Re}\langle \bm{V}, \bm{Y}\rangle_F-\frac{1}{2\gamma}\|\bm{V}\|_F^2\nonumber\\
	&~~{\textrm s.t.}~~ \|(\mathcal{P}^{{\rm Adj}}_{\Omega}(\bm{V}))^H\bm{a}(\theta_k)\|_2\|\bm{\phi}_k\|_2\le 1\nonumber\\
	&\forall \theta_k \in (0,\pi), k=1,..., K 
	%	\begin{bmatrix}
		%		\bm{I}&\bm{V}^H\\
		%		\bm{V}&\bm{Z}_i
		%	\end{bmatrix}\succeq \bm{0},\mathcal{T}^{*}(\bm{Z}_i)=\mathcal{T}^{*}(\bm{I}), i=1,..., K, ~\mathcal{P}_{\Omega}(\bm{V})=\bm{0}	
\end{align}
where $\mathcal{P}^{{\rm Adj}}_{\Omega}$ is the adjoint operator of $\mathcal{P}_{\Omega}$ defined as
\begin{align}
	(\mathcal{P}^{{\rm Adj}}_{\Omega}(\bm{V}))_{(i,l)}=\left\{\begin{array}{cc}
		V_{(i,l)}& i\in\Omega\\
		0& \mathrm{o.w.}
	\end{array}\right\}\forall i,l =1,..., N.
\end{align}
There are $K$ continuous constraints in the above problem, each of which indeed contains infinite number of constraints. This makes the optimization problem highly challenging. In what follows, we state a lemma which converts these infinite constrains into a finite number of linear matrix inequalities, which is tractable using off-the-shell solvers.
\begin{lem}\label{lem.LMI}
	Let $\bm{V}\in\mathbb{C}^{M\times T}$ and $c_1:=\frac{\max_{k=1,..., K}\|\bm{\phi}_k\|_2}{\sqrt{N}}$, then
	\begin{align}\label{contin_constarints}
		&	\|(\mathcal{P}^{{\rm Adj}}_{\Omega}(\bm{V}))^H\bm{a}(\theta_k)\|_2\|\bm{\phi}_k\|_2\le 1~~\forall \theta_k \in (0,\pi), k=1,..., K 
	\end{align} 
	holds if and only if there exists a Hermitian matrix $\bm{Q}\in\mathbb{C}^{N\times N}$ such that
	\begin{align}\label{sdp_constraints}
		&\begin{bmatrix}
			\bm{Q}&\mathcal{P}^{{\rm Adj}}_{\Omega}(\bm{V}) c_1\\
			(\mathcal{P}^{{\rm Adj}}_{\Omega}(\bm{V}))^H c_1&\bm{I}_T
		\end{bmatrix}\succeq \bm{0},\langle \bm{\varTheta}_{q}, \bm{Q}  \rangle=1_{q=0},\nonumber\\
		& q=-N+1,..., N-1,
	\end{align}
	where $\bm{\varTheta}_q$ is the elementary Toeplitz matrix with ones on the $q$-th diagonal and zero elsewhere.
\end{lem}
\begin{proof}
See Appendix \ref{proof.lmi}
\end{proof}

By using this result, the dual problem can be stated in the following equivalent form:
\begin{align}\label{prob.sdp_dual}
	&\max_{\substack{\bm{V}\in\mathbb{C}^{M\times T}\\\bm{Q}\in\mathbb{C}^{N\times N}}}{\rm Re}\langle \bm{V}, \bm{Y}\rangle_F-\frac{1}{2\gamma}\|\bm{V}\|_F^2\nonumber\\
	&\begin{bmatrix}
		\bm{Q}&\mathcal{P}^{{\rm Adj}}_{\Omega}(\bm{V}) c_1\\
		(\mathcal{P}^{{\rm Adj}}_{\Omega}(\bm{V}))^H c_1&\bm{I}_T
	\end{bmatrix}\succeq \bm{0},\nonumber\\
	&\langle \bm{\varTheta}_{q}, \bm{Q}  \rangle=1_{q=0},~	
	q=-N+1,..., N-1.
\end{align}
This problem is a semidefinite programming. Now, we can find the dual of \eqref{prob.sdp_dual} to obtain a problem that is suitable for applying a low complexity ADMM algorithm. First, we write the Lagrange dual function regarding the problem \eqref{prob.sdp_dual} 
\begin{align}\label{lag_dual_fun}
	&\mathcal{L}=\inf_{\bm{V},\bm{Q}}-{\rm Re}\langle \bm{V}, \bm{Y}\rangle_F+\frac{1}{2\gamma}\|\bm{V}\|_F^2-v_1(1-\langle \bm{\varTheta}_0,\bm{Q}\rangle)+\nonumber\\
	&+\sum_{q=2}^{N} v_q\langle\bm{\varTheta}_{q-1},\bm{Q} \rangle+\sum_{q=2}^{N} \overline{v}_{q}\langle\bm{\varTheta}_{-q+1},\bm{Q} \rangle\nonumber\\
	&-\Bigg\langle\begin{bmatrix}
		\bm{\Gamma}&\bm{Z}\\
		\bm{Z}^H&\bm{W}
	\end{bmatrix}, \begin{bmatrix}
		\bm{Q}&\mathcal{P}^{{\rm Adj}}_{\Omega}(\bm{V}) c_1\\
		(\mathcal{P}^{{\rm Adj}}_{\Omega}(\bm{V}))^H c_1&\bm{I}_T
	\end{bmatrix} \Bigg\rangle,
\end{align}
where $\bm{W}\in\mathbb{C}^{T\times T}$, $\bm{\Gamma}\in\mathbb{C}^{N\times N}$, $\bm{Z}\in\mathbb{C}^{N\times T}$ and $\bm{v}\in\mathbb{C}^{N}$ are dual multipliers corresponding to the problem \eqref{prob.sdp_dual}. To prevent \eqref{lag_dual_fun} from getting unbounded, it immediately follows that
\begin{align}
	\bm{\Gamma}=v_1\bm{\varTheta}_0+\sum_{q=2}^{N}v_q\bm{\varTheta}_{q-1}+\overline{v}_q\bm{\varTheta}_{-q+1}=\mathcal{T}(\bm{v})
\end{align}  
where the last part comes from \cite[Equations 2.25, 2.26]{dumitrescu2017positive}. By taking derivative of $\mathcal{L}$ with respect to $\bm{V}$, we have that
\begin{align}
	\widehat{\bm{V}}=\gamma (\bm{Y}+2 c_1 \mathcal{P}_{\Omega}(\bm{Z}))
\end{align} 
which by replacing it into the objective function and forming the dual problem of \eqref{prob.sdp_dual}, we achieve \eqref{prob.goal_optimization}. To identify the angles corresponding to active users, we leverage an adapted version of the results in \cite[Equation 2.2]{candes2014towards} and \cite[Proposition 1]{chi2016guaranteed} (which is indeed for a single device) to conclude that the angles corresponding to the $k$-th active user can be identified by finding angles for which the $\ell_2$ norm of the dual polynomial $\bm{q}_k(\theta)=(\mathcal{P}^{{\rm Adj}}_{\Omega}(\bm{V}))^H\bm{a}(\theta)\|\bm{\phi}_k\|_2$ achieves one. Extending this result to all active users, we can identify all of the angles corresponding to all of the active users by finding locations where $\|\bm{q}_G(\theta)\|_2$ in \eqref{eq.angle_find} maximizes.
%\begin{align}
%\widehat{ \theta}_k=\mathop{\arg\max}\|q\|_2
%\end{align}
%
%
% we can now find the single dual polynomial $\bm{q}$To identify active devices, we use the following proposition:
%\begin{prop}
%Denote the set of angles corresponding to $k$-th device by $\mathcal{S}^a_k$. If there exists multiple vector-valued polynomials $\bm{q}_k(\theta):=(\mathcal{P}^{Adj}_{\Omega}(\bm{V}))^H\bm{a}(\theta), \theta\in\mathcal{S}^a_k$ satisfying
%\begin{align}
%\|\bm{q}_G(\theta)\|_2< \frac{1}{\max_{k=1,..., K}\|\bm{\phi}_k\|_2} \forall \theta \in 
%\end{align}
%\end{prop}
\section{Derivation of the dual problem}\label{proof.dualprob}
By writing the Lagrangian function of the convex optimization problem \eqref{prob.atomic_l1_lasso}, we have:
\begin{align}\label{d1}
& \mathcal{L}(\bm{Z}_1,...,\bm{Z}_K,\bm{Y}^{\star},\bm{V})=\sum_{k=1}^K\|\bm{Z}_k\|_{\mathcal{A}_k}+\frac{\gamma}{2}\|\bm{Y}-\bm{Y}^{\star}\|_F^2+\nonumber\\
&{\rm Re}~\langle \bm{Y}^{\star}-\sum_{k=1}^K \mathcal{P}_{\Omega}(\bm{Z}_k),\bm{V}\rangle.
\end{align}
Minimizing $\mathcal{L}$ in \eqref{d1} with respect to $\bm{Z}_k$s and $\bm{Y}^{\star}$ gives the following:
\begin{align}\label{d2}
&\inf_{\{\bm{Z}_k\}_{k=1}^K}\sum_{k=1}^K\Big[\|\bm{Z}_k\|_{\mathcal{A}_k}-{\rm Re}\langle\mathcal{P}_{\Omega}(\bm{Z}_k),\bm
V \rangle\Big]+\nonumber\\
&\inf_{\bm{Y}^{\star}}\frac{\gamma}{2}\|\bm{Y}-\bm{Y}^{\star}\|_F^2+{\rm Re}~\langle \bm{Y}^{\star},\bm{V}\rangle\stackrel{(\RN{1})}{=}\nonumber\\
&\inf_{\{\bm{Z}_k\}_{k=1}^K} \sum_{k=1}^K\|\bm{Z}_k\|_{\mathcal{A}_k}(1-\|\mathcal{P}^{{\rm Adj}}_{\Omega}(\bm{V})\|^{d}_{\mathcal{A}_k})+\nonumber\\
&\inf_{\bm{Y}^{\star}}\frac{\gamma}{2}\|\bm{Y}-\bm{Y}^{\star}\|_F^2+{\rm Re}\langle \bm{Y}^{\star},\bm{V}\rangle\stackrel{(\RN{2})}{=}\inf_{\bm{Y}^{\star}}\frac{\gamma}{2}\|\bm{Y}-\bm{Y}^{\star}\|_F^2+\nonumber\\
&{\rm Re}\langle \bm{Y}^{\star},\bm{V}\rangle={\rm Re}~\langle \bm{Y},\bm{V}\rangle -\frac{1}{2\gamma}\|\bm{V}\|_F^2
\end{align}
where in the first optimization in $(\RN{1})$, we used the relation 
\begin{align}
Re \langle \mathcal{P}_{\Omega}(\bm{Z}_k), \bm{V} \rangle=Re \langle \bm{Z}_k, \mathcal{P}^{{\rm Adj}}_{\Omega}(\bm{V}) \rangle \le 
\|\bm{Z}_k\|_{\mathcal{A}_k}\|\mathcal{P}^{{\rm Adj}}_{\Omega}(\bm{V}) \|^d_{\mathcal{A}_k},
\end{align}
which comes from the H\"{o}lder's inequality. As the upper bound in the H\"{o}lder's inequality is achievable, the infimum in the first term of $(\RN{1})$ reaches its lower bound. Here, 
$\|\cdot\|^d_{\mathcal{A}_k}$ is the dual norm associated with $\|\cdot\|_{\mathcal{A}_k}$ defined as
\begin{align}\label{d3}
&\|\mathcal{P}^{{\rm Adj}}_{\Omega}(\bm{V})\|^d_{\mathcal{A}_k}:=\sup_{\|\bm{Z}\|_{\mathcal{A}_k}\le 1}\langle \mathcal{P}^{{\rm Adj}}_{\Omega}(\bm{V}), \bm{Z}\rangle 
=\sup_{\theta_k\in (0,\pi),\bm{\phi}_k}\nonumber\\
&\langle \mathcal{P}^{{\rm Adj}}_{\Omega}(\bm{V}), \bm{a}(\theta_k)\bm{\phi}^H_k\rangle=\sup_{\theta_k\in (0,\pi),\bm{\phi}_k}\langle (\mathcal{P}^{{\rm Adj}}_{\Omega}(\bm{V}))^H\bm{a}(\theta_k),\bm{\phi}_k\rangle\nonumber\\
&=\sup_{\theta_k\in (0,\pi),\bm{\phi}_k}\|(\mathcal{P}^{{\rm Adj}}_{\Omega}(\bm{V}))^H\bm{a}(\theta_k)\|_2\|\bm{\phi}_k\|_2
\end{align}
where in the last step above, we used again H\"{o}lder's inequality.

In order to have a bounded objective function in \eqref{d2}, we should have $\|\mathcal{P}^{{\rm Adj}}_{\Omega}(\bm{V})\|^{d}_{\mathcal{A}_k}\le 1$ for all $k=1,..., K$ making the objective function in the first optimization in \eqref{d2} equal to zero, thus leading to the equality $(\RN{2})$. By minimizing the objective function in $(\RN{2})$ with respect to $\bm{Y}^{\star}$, we achieve $\bm{Y}^{\star}=\bm{Y}-\frac{\bm{V}}{\gamma}$ and after replacing into the objective function \eqref{d2}, the Lagrangian function reads as $\langle \bm{Y},\bm{V}\rangle -\frac{1}{2\gamma}\|\bm{V}\|_F^2$. Thus, the resulting dual problem becomes in the form of the following optimization:
\begin{align}\label{d4}
&\max_{\bm{V}\in\mathbb{C}^{N\times T}}{\rm Re}\langle \bm{V}, \bm{Y}\rangle_F-\frac{1}{2\gamma}\|\bm{V}\|_F^2\nonumber\\
&{\rm s.t.}~~\|\mathcal{P}^{{\rm Adj}}_{\Omega}(\bm{V})\|^d_{\mathcal{A}_k}\le 1,~k=1,..., K.
\end{align}
Combining \eqref{d3} and \eqref{d4}, leads to
\eqref{prob.dual}. 

\section{Proof of Lemma \ref{lem.LMI}}\label{proof.lmi}
First, we begin with the fact that the constraints 
\begin{align}
	&	\|(\mathcal{P}^{{\rm Adj}}_{\Omega}(\bm{V}))^H\bm{a}(\theta_k)\|_2\|\bm{\phi}_k\|_2\le 1~\forall \theta_k \in (0,\pi), k=1,..., K 
\end{align} 
are equivalent to 
\begin{align}
	&	\|(\mathcal{P}^{{\rm Adj}}_{\Omega}(\bm{V}))^H\bm{a}(\theta_k)\|_2\max_{k=1,..., K}\|\bm{\phi}_k\|_2\le 1~\forall \theta_k \in (0,\pi), 
\end{align} 
which then by defining the notation $c_1$ and \eqref{eq.atoms} can be rewritten as follows:
\begin{align}\label{pl1}
	\sum_{i=1}^T \Big|\sum_{l=1}^N (\mathcal{P}^{{\rm Adj}}_{\Omega}(\bm{V}))_{(l,i)} c_1 {\rm e}^{-j 2\pi (l-1) \cos(\theta)}\Big|^2\le 1,~~\forall \theta\in (0,\pi).
\end{align}
Defining the $i$-th column of $\mathcal{P}^{{\rm Adj}}_{\Omega}(\bm{V})$ as $\bm{a}_i \coloneqq c_1[(\mathcal{P}^{{\rm Adj}}_{\Omega}(\bm{V}))_{(1,i)},..., (\mathcal{P}^{{\rm Adj}}_{\Omega}(\bm{V}))_{(N,i)}]^T$ and $\bm{f}(\theta) \coloneqq [1,..., {\rm e}^{j 2\pi (N-1) \cos(\theta)}]^T$, \eqref{pl1} can be reformulated as
\begin{align}
	\sum_{i=1}^T|\bm{f}(\theta)^H\bm{a}_i|^2\le 1,
\end{align} 
which implies that the polynomial $1-\bm{f}(\theta)^H\sum_{i=1}^T\bm{a}_i\bm{a}_i^H\bm{f}(\theta)$ is non-negative. Based on \cite[Theorem 1.1]{dumitrescu2017positive}, there exists a polynomial $\bm{q}_1(\theta):=\sum_{i=1}^N b_i{\rm e}^{-j 2\pi (i-1) \cos(\theta)}=\bm{f}(\theta)^H\bm{b}$ for some $\bm{b}\in\mathbb{C}^{N}$ such that
\begin{align}\label{pl2}
	1-\bm{f}(\theta)^H\sum_{i=1}^T\bm{a}_i\bm{a}_i^H\bm{f}(\theta)=|q_1(\theta)|^2=\bm{f}(\theta)^H\bm{b}\bm{b}^H\bm{f}(\theta).
\end{align}
Set $\bm{Q}=\sum_{i=1}^T\bm{a}_i\bm{a}_i^H+\bm{b}\bm{b}^H$. Since both $\bm{Q}$ and $\bm{Q}-\sum_{i=1}^T\bm{a}_i\bm{a}_i^H$ are positive semidefinite, using Schur complement lemma \cite[Appendix A.5.5]{boyd2004convex} leads to the semidefinite constraint
\begin{align}
	\begin{bmatrix}
		\bm{Q}& \mathcal{P}^{{\rm Adj}}_{\Omega}(\bm{V}) c_1\\
		(\mathcal{P}^{{\rm Adj}}_{\Omega}(\bm{V}))^H c_1&\bm{I}_T
	\end{bmatrix}\succeq \bm{0}. 
\end{align} 
Further, by \eqref{pl2}, we have:
\begin{align}
	\bm{f}^H(\theta)\bm{Q}\bm{f}(\theta)=\langle \bm{Q}, \bm{f}(\theta)\bm{f}^H(\theta) \rangle=1, \forall \theta \in (0,\pi).
\end{align} 
This can be further simplified to
\begin{align}\label{pl3}
	&\langle \bm{Q}, \mathcal{T}(\bm{f}(\theta))\rangle=\langle \mathcal{T}^{{\rm Adj}}(\bm{Q}), \bm{f}^R(\theta)\rangle-j\langle \mathcal{T}^{{\rm Adj}}(\bm{C}\odot \bm{Q}), \bm{f}^I(\theta)\rangle\nonumber\\
	&~\forall \theta\in (0,\pi)
\end{align}
where $\bm{f}^R(\theta)$ and $\bm{f}^I(\theta)$ are the real and imaginary parts of $\bm{f}(\theta)$, respectively. $\bm{C}$ is a matrix composed of $-1$s on the lower triangular elements and $1$s elsewhere, i.e., $\bm{C}_{k,l}=-1,k<l$ and $\bm{C}_{k,l}=1,k\ge l$. Satisfying \eqref{pl3} for all $\theta$ is only possible when
\begin{align}
	&(\mathcal{T}^{{\rm Adj}}(\bm{Q}))_1=1\label{pl41}\\
	&(\mathcal{T}^{{\rm Adj}}(\bm{Q}))_k=0,\label{pl42}\\ 
	&(\mathcal{T}^{{\rm Adj}}(\bm{C}\odot\bm{Q}))_k=0\label{pl43}\\
	&k=2,\ldots, N.
\end{align}
Using the result of Lemma \ref{lem.adjointoperator_calculation}, we have that \eqref{pl41} simplifies to
\begin{align}\label{pl5}
	\sum_{i=1}^N Q(i,i)=1.
\end{align}
For \eqref{pl42} and \eqref{pl43} to hold, we should have
\begin{align}\label{pl7}
	\sum_{i=k}^N {\rm Re}~ Q(i-k+1,i)=0
\end{align}
and 
\begin{align}
	\sum_{i=k}^N Q(i-k+1,i)=\sum_{i=k}^N Q(i,i-k+1),	
\end{align} 
which leads to
\begin{align}\label{pl6}
	\sum_{i=k}^N {\rm Im}~ Q(i-k+1,i)=0.
\end{align}
Combining \eqref{pl6} and \eqref{pl7} gives $\sum_{i=k}^N Q(i-k+1,i)=\sum_{i=k}^N Q(i,i-k+1)=0$, which along with the relation \eqref{pl5} could be all simply written as
\begin{align}\label{pl8}
	\langle \bm{\varTheta}_{q},\bm{Q}\rangle= 1_{q=0}, ~q=-N+1,..., N-1.	
\end{align}

Conversely, if \eqref{sdp_constraints} holds, then by Schur complement lemma \cite[Appendix A.5.5]{boyd2004convex}, we have 
\begin{align}\label{pl9}
	\bm{Q}-\sum_{i=1}^T \bm{a}_i\bm{a}_i^H\succeq \bm{0}
\end{align}
and \eqref{pl8}. Thus, by \eqref{pl9}, we may write
\begin{align}\label{pl10}
	&\sum_{i=1}^T|\bm{f}(\theta)^H\bm{a}_i|^2=\bm{f}(\theta)^H\Big(\sum_{i=1}^T\bm{a}_i\bm{a}_i^H\Big)\bm{f}(\theta)\le \bm{f}(\theta)^H\bm{Q}\bm{f}(\theta)\nonumber\\
	&=\langle \bm{Q}, \bm{f}(\theta)\bm{f}(\theta)^H\rangle=\langle \bm{Q}, \mathcal{T}(\bm{f}(\theta)) \rangle.
\end{align}
Using \eqref{eq.toep_rel_complex}, we can proceed \eqref{pl10} by writing
\begin{align}\label{pl11}
	\langle \bm{Q}, \mathcal{T}(\bm{f}(\theta)) \rangle=\langle \mathcal{T}^{{\rm Adj}}(\bm{Q}), \bm{f}^R(\theta)\rangle-j\langle \mathcal{T}^{{\rm Adj}}(\bm{C}\odot \bm{Q}),\bm{f}^I(\theta)\rangle.
\end{align}
Due to the relation $\langle \bm{\varTheta}_{q}, \bm{Q}  \rangle=1_{q=0}, q=-N+1,..., N-1$, it follows that 
\begin{align}
	&(\mathcal{T}^{{\rm Adj}}(\bm{Q}))_1=1\\
	&(\mathcal{T}^{{\rm Adj}}(\bm{Q}))_k=0, k=2,\ldots, N\\
	&(\mathcal{T}^{{\rm Adj}}(\bm{C}\odot\bm{Q}))_k=0, k=2,\ldots, N,
\end{align}
which gives $\langle \bm{Q}, \mathcal{T}(\bm{f}(\theta)) \rangle=1$ by \eqref{pl3}. Thus, we have $\sum_{i=1}^T|\bm{f}(\theta)^H\bm{a}_i|^2\le 1$, which is equivalent to \eqref{contin_constarints}.

\section{Proof of Lemmas \ref{lem.adjointoperator_calculation} and \ref{lem.adjointoper_cal_lower_tringular}}\label{proof.lem.adjoint_operator}
For an arbitrary complex-valued vector $\bm{v}\in\mathbb{C}^N$, the Toeplitz matrix $\mathcal{T}(\bm{v})$ defined in \eqref{eq.toeplitz_mat} can be splitted into its real and imaginary components as
\begin{align}\label{eq.toep_rel_complex}
	\mathcal{T}(\bm{v})=\mathcal{T}(\bm{v}^{R})+j \bm{C}\odot \mathcal{T}(\bm{v}^I).	
\end{align}
As it is observed from the latter relation, the Toeplitz operator $\mathcal{T}:\mathbb{C}^N\rightarrow \mathbb{C}^{N\times N}$ is not a linear operator in general. Thus, obtaining the adjoint operator requires splitting the input space of the Toeplitz operator into real and complex-valued, in which cases it is linear and we can obtain its adjoint operator.  Therefore, for an arbitrary complex-valued matrix $\bm{A}\in\mathbb{C}^{N\times N}$ and the real-valued vector $\bm{v}^R\in\mathbb{R}^{N}$, the adjoint operator $\mathcal{T}^*$ is obtained from the following formula
\begin{align}\label{eq.main_formula}
	\langle \mathcal{T}(\bm{v}^R),\bm{A}\rangle= \langle \bm{v}^R,\mathcal{T}^{{\rm Adj}}(\bm{A})\rangle. 
\end{align}
The left-hand side of \eqref{eq.main_formula} can be written as
\begin{align}\label{eq.r1}
	&	\sum_{\substack{i,l=1\\i\le l}}^N \overline{A}_{i,l} v^R_{|i-l|+1}+\sum_{\substack{i,l=1\\i> l}}^N \overline{A}_{i,l} \overline{v}^R_{|i-l|+1}\stackrel{(a)}{=}\nonumber\\
	&\sum_{i=1}^N\sum_{k=\max\{1,i-N+1\}}^i \overline{A}_{i,i-k+1} v_{k}+\nonumber\\
	&\sum_{l=1}^N\sum_{q=\max\{1,i-N+1\}}^i \overline{A}_{l-q+1,l} v^R_{q}\stackrel{(b)}{=}\nonumber\\
	&v^R_1\sum_{i=1}^N \overline{A}_{i,i}+\sum_{k=2}^Nv^R_{k}\Big(\sum_{i=k}^N \overline{A}_{i,i-k+1}+\overline{A}_{i-k+1,i}\Big) 
\end{align}
where in (a) we use change of variables $k=i-l+1$ and $q=l-i+1$, for the first and second terms, respectively, and in (b), we change the order of summation and determine the lower and upper limit of summation correspondingly. Combining \eqref{eq.main_formula} and \eqref{eq.r1}, we get the following result
\begin{align}
	&(\mathcal{T}^{\rm {\rm Adj}}(\bm{A}))_1=\sum_{i=k}^N {A}_{i,i}\nonumber\\	
	&(\mathcal{T}^{\rm {\rm Adj}}(\bm{A}))_k=\sum_{i=k}^N {A}_{i,i-k+1}+{A}_{i-k+1,i},~ k=2:N.
\end{align}
The proof of Lemma \ref{lem.adjointoper_cal_lower_tringular} proceeds in a similar manner. First, by the definition of the operator $\mathcal{T}_1$, it follows that
\begin{align}\label{eq.toeplitz1_mat}
	\mathcal{T}_1(\bm{v})=\begin{bmatrix}
		v_1&v_2&\hdots&v_{N}\\
		-\overline{v}_2&v_1&\hdots&v_{N-1}\\
		\vdots&\vdots&\ddots&\vdots\\
		-\overline{v}_N&-\overline{v}_{N-1}&\hdots&v_1
	\end{bmatrix}.
\end{align}
Using the relation  \eqref{eq.main_formula} for $\mathcal{T}_1$, we can derive the adjoint operator for arbitrary real-valued vector $\bm{v}^I$ by using the following relation
\begin{align}\label{eq.main_formula1}
	\langle \mathcal{T}_1(\bm{v}^I),\bm{A}\rangle= \langle \bm{v}^I,\mathcal{T}_1^{{\rm Adj}}(\bm{A})\rangle. 
\end{align}
Proceeding with the left-hand side and leveraging the same reasons as in \eqref{eq.r1}, leads to
\begin{align}\label{eq.rr1}
	&	\sum_{\substack{i,l=1\\i\le l}}^N \overline{A}_{i,l} v^I_{|i-l|+1}-\sum_{\substack{i,l=1\\i> l}}^N \overline{A}_{i,l} {v}^I_{|i-l|+1}\stackrel{(\RN{1})}{=}\nonumber\\
	&\sum_{i=1}^N\sum_{k=\max\{1,i-N+1\}}^i \overline{A}_{i,i-k+1} v^I_{k}-\nonumber\\
	&\sum_{l=1}^N\sum_{q=\max\{1,i-N+1\}}^i \overline{A}_{l-q+1,l} v^I_{q}\stackrel{(\RN{2})}{=}\nonumber\\
	&v^I_1\sum_{i=1}^N \overline{A}_{i,i}+\sum_{k=2}^Nv^I_{k}\Big(\sum_{i=k}^N \overline{A}_{i,i-k+1}-\overline{A}_{i-k+1,i}\Big) 
\end{align}
where a change of summation order is used in the last part. Combining \eqref{eq.rr1} and \eqref{eq.main_formula1}, we reach to the result of \eqref{eq.T1adj}.

\section{Proof of Lemma \ref{lem.v_cal}}\label{proof.lem.v_cal}
Define $\bm{v}:=\bm{v}^R+j\bm{v}^I$. 
To prove the result, we first rewrite the Lagrangian function $\mathcal{L}_{\rho}$ in \eqref{eq.L_rho} and only keep the terms related to $\bm{v}$ as other terms do not affect the minimization over $v$. Thus, we have
\begin{align}
	\mathcal{L}_{\rho}(\bm{v})=v_1^R-{\rm Re}\langle \bm{\Lambda}_0, \mathcal{T}(\bm{v}) \rangle+\frac{\rho}{2}\|\bm{\Psi}_0-\mathcal{T}(\bm{v})\|_F^2.
\end{align}
By splitting into real and imaginary parts and using \eqref{eq.toep_rel_complex}, it follows that
\begin{align}
	&\mathcal{L}_{\rho}(\bm{v})=v_1^R-{\rm Re}\bigg(\langle \mathcal{T}^{{\rm Adj}}(\bm{\Lambda}_0), \bm{v}^R\rangle\bigg)+\nonumber\\
	&{\rm Re}\bigg( j \langle \mathcal{T}^{{\rm Adj}}(\bm{C}\odot \bm{\Lambda}_0), \bm{v}^I\rangle\bigg)+\frac{\rho}{2}\|\bm{\Psi}^R_0-\mathcal{T}(\bm{v}^R)\|_F^2+\nonumber\\
	&\frac{\rho}{2}\|\bm{\Psi}^I_0-\underbrace{\bm{C}\odot\mathcal{T}(\bm{v}^I)}_{=:\mathcal{T}_1^{{\rm Adj}}(\bm{v}^I)}\|_F^2,
\end{align}
which can be further simplified to
\begin{align}\label{rel_v1}
	&\mathcal{L}_{\rho}(\bm{v})=v_1^R-\langle \mathcal{T}^{{\rm Adj}}(\bm{\Lambda}_0^R), \bm{v}^R\rangle-\nonumber\\
	& \langle \mathcal(\bm{C}\odot \bm{\Lambda}_0^R), \bm{v}^I\rangle+\frac{\rho}{2}\|\bm{\Psi}^R_0-\mathcal{T}(\bm{v}^R)\|_F^2+\nonumber\\
	&\frac{\rho}{2}\|\bm{\Psi}^I_0-\underbrace{\bm{C}\odot\mathcal{T}(\bm{v}^I)}_{=:\mathcal{T}_1^{{\rm Adj}}(\bm{v}^I)}\|_F^2.
\end{align}
Taking derivative of $\mathcal{L}_{\rho}$ with respect to $\bm{v}^R$, we have
\begin{align}\label{rel_v2}
	\frac{\partial \mathcal{L}_{\rho}}{\partial \bm{v}^R}=\bm{e}_1-\mathcal{T}^{{\rm Adj}}(\bm{\Lambda}_0^R)-\rho \mathcal{T}^{{\rm Adj}}(\bm{\Psi}_0^R-\mathcal{T}(\bm{v}^R))=\bm{0}.
\end{align}
Moreover, taking the derivative of $\mathcal{L}_{\rho}$ with respect to $\bm{v}^I$, we reach to the following expression
\begin{align}\label{rel_v3}
	-\mathcal{T}^{{\rm Adj}}(\bm{C}\odot\bm{\Lambda}_0^I)-\rho \mathcal{T}_1^{{\rm Adj}}(\bm{\Psi}_0^I-\mathcal{T}_1(\bm{v}^I))=\bm{0}.
\end{align}
Simplifying \eqref{rel_v2} and \eqref{rel_v3} requires to characterize $\mathcal{T}^{{\rm Adj}}(\cdot)$ and $\mathcal{T}_1^{{\rm Adj}}(\cdot)$, which is done in the following lemmas and proved in Appendix \ref{proof.lem.adjoint_operator}.
\begin{lem}\label{lem.adjointoperator_calculation}
	Let $\bm{A}\in\mathbb{C}^{N\times N}$ be an arbitrary complex matrix. The adjoint operator of the Toeplitz operator $\mathcal{T}$ denoted by $\mathcal{T}^{{\rm Adj}}: \mathbb{C}^{N\times N}\rightarrow \mathbb{C}$ is obtained by
	\begin{align}\label{T_adj}
		&(\mathcal{T}^{{\rm Adj}}(\bm{A}))_{1}= \sum_{i=1}^N \bm{A}_{i,i},\\
		&(\mathcal{T}^{{\rm Adj}}(\bm{A}))_{k}=\sum_{i=k}^N \big(\bm{A}_{i,i-k+1}+\bm{A}_{i-k+1,i}\big), ~k\neq 1.
	\end{align}	
\end{lem}
\begin{lem}\label{lem.adjointoper_cal_lower_tringular}
	Consider the operator $\mathcal{T}_1(\cdot): \mathbb{C}^{N}\rightarrow \mathbb{C}^{N\times N}$, which for any arbitrary vector $\bm{v}$ is defined as $\mathcal{T}_1(\bm{v}) \coloneqq \bm{C}\circ \mathcal{T}(\bm{v})$ where $\bm{C}\in\mathbb{C}^{N\times N}$ is a matrix composed of $-1$s on lower triangular parts and $1$s elsewhere, i.e., $\bm{C}_{k,l}=-1,k<l$ and $\bm{C}_{k,l}=1,k\ge l$. Then, the adjoint operator of $\mathcal{T}_1$ denoted by $\mathcal{T}_1^{{\rm Adj}}$ for any arbitrary matrix $\bm{A}\in \mathbb{C}^{N\times N}$ is obtained as
	\begin{align}\label{eq.T1adj}
		&(\mathcal{T}_1^{{\rm Adj}}(\bm{A}))_{1}= \sum_{i=1}^N \bm{A}_{i,i},\\
		&(\mathcal{T}_1^{{\rm Adj}}(\bm{A}))_{k}=\sum_{i=k}^N \big(\bm{A}_{i-k+1,i}-\bm{A}_{i,i-k+1}\big), ~k\neq 1.
	\end{align}		
\end{lem}
From Lemmas \ref{lem.adjointoperator_calculation} and \ref{lem.adjointoper_cal_lower_tringular}, we can also obtain the expressions $\mathcal{T}^{{\rm Adj}}(\mathcal{T}(\bm{z}))$ and $\mathcal{T}_1^{{\rm Adj}}(\mathcal{T}_1(\bm{z}))$, which are given in the following corollary and proved in Appendix \ref{proof.corol}.
\begin{corl}\label{corol}
	Consider the operators $\mathcal{T}(\cdot)$ and $\mathcal{T}_1(\cdot)$ which are defined in \eqref{eq.toeplitz_mat} and \eqref{eq.toeplitz1_mat}, respectively. Then, for an arbitrary $\bm{z}\in\mathbb{C}^N$, the following relations hold
	\begin{align}
		&\mathcal{T}^{{\rm Adj}}(\mathcal{T}(\bm{z}))_1=\mathcal{T}_1^{{\rm Adj}}(\mathcal{T}_1(\bm{z}))_1=N v_1\nonumber\\
		&\mathcal{T}^{{\rm Adj}}(\mathcal{T}(\bm{z}))_k=\mathcal{T}_1^{{\rm Adj}}(\mathcal{T}_1(\bm{z}))_k=2 {\rm Re}(v_k) (N-k+1).
	\end{align} 	
\end{corl}
Leveraging the results provided in Lemmas \ref{lem.adjointoperator_calculation}, \ref{lem.adjointoper_cal_lower_tringular} and Corollary \ref{corol}, we can proceed \eqref{rel_v2} and \eqref{rel_v3} as follows
\begin{align}
	\bm{v}^{R}=\bm{g} \odot \bigg(-\frac{\bm{e}_1}{\rho}+\mathcal{T}^{{\rm Adj}}(\bm{\Psi}_0^R+\frac{\bm{\Lambda}_0^R}{\rho})\bigg)
\end{align}
and 
\begin{align}
	\bm{v}^{I}=\bm{g} \odot \bigg(\mathcal{T}_1^{{\rm Adj}}(\bm{\Psi}_0^I)+\frac{\mathcal{T}^{{\rm Adj}}(\bm{C}\odot\bm{\Lambda}_0^I)}{\rho}\bigg),
\end{align}
which can be combined to achieve
\begin{align}
	&\bm{v}=\bm{g}\odot \Bigg(-\frac{\bm{e}_1}{\rho}+\mathcal{T}^{{\rm Adj}}(\bm{\Psi}^{R}_0)+\mathcal{T}_1^{{\rm Adj}}(\bm{\Psi}^{I}_0)+\nonumber\\
	&\frac{\mathcal{T}^{{\rm Adj}}(\bm{\Lambda}_0^R+\bm{C}\odot \bm{\Lambda}_0^I)}{\rho}\Bigg).
\end{align}
% \section{Proof of Lemma \ref{lem.adjointoper_cal_lower_tringular}}\label{proof.lem.adjoint_operator1}

\section{Proof of corollaries }\label{proof.corol}
Using the relations in Lemmas \ref{lem.adjointoperator_calculation} and \ref{lem.adjointoper_cal_lower_tringular}, i.e., \eqref{T_adj} and \eqref{eq.T1adj}, we have that for an arbitrary vector $\bm{v}\in\mathbb{C}^{N}$
\begin{align}
	&(\mathcal{T}^{{\rm Adj}}\mathcal{T}(\bm{v}))_1=\sum_{i=1}^N \mathcal{T}(\bm{v})_{(i,i)}=\sum_{i=1}^N v_1=Nv_1,\\
	&(\mathcal{T}^{{\rm Adj}}\mathcal{T}(\bm{v}))_k=\sum_{i=1}^N \mathcal{T}(\bm{v})_{(i,i-k+1)}+\mathcal{T}(\bm{v})_{(i-k+1,i)}=\nonumber\\
	&\sum_{i=k}^N v_k+\overline{v}_k=2 {\rm Re}(v_k) (N-k+1).
\end{align}
Moreover, for $\mathcal{T}_1^{{\rm Adj}}(\mathcal{T}_1(\bm{v}))$, we have
\begin{align}
	&(\mathcal{T}_1^{{\rm Adj}}\mathcal{T}(\bm{v}))_1=\sum_{i=1}^N \mathcal{T}(\bm{v})_{(i,i)}=\sum_{i=1}^N v_1=Nv_1,\\
	&(\mathcal{T}^{{\rm Adj}}\mathcal{T}(\bm{v}))_k=\sum_{i=1}^N \mathcal{T}(\bm{v})_{(i,i-k+1)}-\mathcal{T}(\bm{v})_{(i-k+1,i)}=\nonumber\\
	&\sum_{i=k}^N v_k-(-\overline{v}_k)=2 {\rm Re}(v_k) (N-k+1).
\end{align}
\section{Proof of Lemma \ref{lem.dual_var+dual_multiplier}}\label{proof.lem.dual_rel}
To prove the result, we first borrow a useful lemma from \cite{tang2013compressed}, which states that for any $\bm{Z}\in\mathbb{C}^{N\times T}$ 

\begin{align}
&\|\bm{Z}\|_{\mathcal{A}_{\star}}=\min_{\bm{v}\in\mathbb{C}^{N},\bm{E}\in\mathbb{C}^{T\times T}} {\rm Re}\Big(\frac{{v}_1}{2}+ \frac{{\rm tr}(\bm{E})}{2}\Big)\\
&\begin{bmatrix}
\mathcal{T}(\bm{v})& \bm{Z}\\
\bm{Z}^H&\bm{E}
\end{bmatrix}\succeq \bm{0}.
\end{align}
Suppose that $\widehat{\bm{Z}}$ is the optimal solution of the goal-oriented optimization problem 
in \eqref{prob.goal_optimization}, which is alternatively the optimal solution of
\begin{align}\label{e6}
\min_{\bm{Z}}2\|\bm{Z}\|_{\mathcal{A}_{\star}}+\frac{\gamma}{2}\|\bm{Y}+2c_1\mathcal{P}_{\Omega}(\bm{Z})\|_F^2=:J(\bm{Z}).
\end{align}
Since the objective function is convex, we must have $\bm{0}\in\partial J(\widehat{\bm{Z}})$ which leads to
\begin{align}\label{e1}
\partial \|\cdot\|_{\mathcal{A}_{\star}}(\widehat{\bm{Z}})=-\gamma c_1 \mathcal{P}_{\Omega}^{{\rm Adj}}(\bm{Y}+2c_1 \mathcal{P}_{\Omega}(\widehat{\bm{Z}})). 
\end{align}
Also, the definition of subdifferential function imposes the following relation for any arbitrary $\bm{Z}$:
\begin{align}\label{e2}
\|\bm{Z}\|_{\mathcal{A}_{\star}}\ge \|\widehat{\bm{Z}}\|_{\mathcal{A}_{\star}}+\langle \bm{Z}-\widehat{\bm{Z}}, -\gamma c_1 \mathcal{P}_{\Omega}^{{\rm Adj}}(\bm{Y}+2c_1 \mathcal{P}_{\Omega}(\widehat{\bm{Z}})) \rangle.
\end{align}
This also implies that
\begin{align}\label{e3}
&\inf_{\bm{Z}} \Bigg[\|\bm{Z}\|_{\mathcal{A}_{\star}}+\langle \bm{Z}, \gamma c_1 \mathcal{P}_{\Omega}^{{\rm Adj}}(\bm{Y}+2c_1 \mathcal{P}_{\Omega}(\bm{Z})) \rangle\Bigg] \ge\nonumber\\
& \|\widehat{\bm{Z}}\|_{\mathcal{A}_{\star}}+\langle \widehat{\bm{Z}}, \gamma c_1 \mathcal{P}_{\Omega}^{{\rm Adj}}(\bm{Y}+2c_1 \mathcal{P}_{\Omega}(\widehat{\bm{Z}})) \rangle.
\end{align}
The minimization problem in the relation \eqref{e3} remains bounded only when
\begin{align}
\|\gamma c_1 \mathcal{P}_{\Omega}^{{\rm Adj}}(\bm{Y}+2c_1 \mathcal{P}_{\Omega}(\bm{Z})) \rangle\|^{d}_{\mathcal{A}_{\star}}\le 1
\end{align}
which indeed leads to
\begin{align}\label{e5}
\|\widehat{\bm{Z}}\|_{\mathcal{A}_{\star}}=-\langle \widehat{\bm{Z}}, \gamma c_1 \mathcal{P}_{\Omega}^{{\rm Adj}}(\bm{Y}+2c_1 \mathcal{P}_{\Omega}(\widehat{\bm{Z}})) \rangle.
\end{align}
By replacing \eqref{e5} into $J(\cdot)$ function in \eqref{e6}, we have that
\begin{align}\label{e7}
&J(\widehat{ \bm{Z}})=\frac{\gamma}{2}\|\bm{Y}+2c_1\mathcal{P}_{\Omega}(\widehat{\bm{Z}})\|_F^2-2\gamma c_1 \Big\langle \mathcal{P}_{\Omega}(\widehat{\bm{Z}}), \nonumber\\
&\bm{Y}+2c_1\mathcal{P}_{\Omega}(\widehat{\bm{Z}}) \Big\rangle=\gamma \Big\langle \bm{Y}, \bm{Y}+2c_1\mathcal{P}_{\Omega}(\widehat{\bm{Z}}) \Big\rangle\nonumber\\
&-\frac{\gamma}{2}\|\bm{Y}+2c_1\mathcal{P}_{\Omega}(\widehat{\bm{Z}})\|_F^2.
\end{align}
By leveraging strong duality in convex optimization, the objective function of the primal \eqref{e6} and dual \eqref{prob.dual} optimizations must be equal. Hence, by combining \eqref{e7} and \eqref{prob.dual}, we can find out that
\begin{align}\label{e8}
\widehat{\bm{V}}=\gamma(\bm{Y}+2c_1 \mathcal{P}_{\Omega}(\bm{Z})).
\end{align}
In addition, by minimizing the Lagrangian function $\mathcal{L}_{0}$ and ignoring the $\rho$ term (as $\rho$ is only for ADMM algorithm), we have $2 \gamma c_1\mathcal{P}_{\Omega}^{{\rm Adj}}(\bm{Y}+2c_1 \mathcal{P}_{\Omega}(\bm{Z}))-2\bm{\Lambda}_1=\bm{0}$, which by \eqref{e8} leads to 
%
%combining \eqref{e8} and \eqref{diff_Lrho_Z}, we have that
\begin{align}
2 c_1\mathcal{P}_{\Omega}^{{\rm Adj}}(\widehat{\bm{V}})-2\bm{\Lambda}_1=0
\end{align}
and proves the final result.

\ifCLASSOPTIONcaptionsoff

\fi
\bibliographystyle{ieeetr}
\bibliography{references1}
\end{document}